\documentclass[10pt,letterpaper,twocolumn]{article}
\usepackage[T1]{fontenc}
\usepackage[utf8]{inputenc}
\usepackage{lmodern}
\usepackage[margin=0.75in]{geometry}
\usepackage{amsmath,amssymb,amsfonts,bm,amsthm,mathtools}
\usepackage{graphicx,booktabs,microtype,placeins,xcolor}
\usepackage[numbers,sort&compress]{natbib}
\usepackage{authblk,titlesec}
\usepackage[hidelinks]{hyperref}
\hypersetup{
 pdftitle={Wavelength-Uniform Quantum Algorithms for Mixed-State Quantum Dynamics},
 pdfauthor={Shi Jin and Chuwen Ma}
}
\graphicspath{{figures/}}
\titleformat{\section}{\centering\small\bfseries}{\thesection.}{1em}{\MakeUppercase}
\titleformat{\subsection}{\centering\small\bfseries}{\thesection.\thesubsection.}{1em}{}
\titleformat{\paragraph}[runin]{\normalfont\itshape}{}{0pt}{}
\titlespacing*{\section}{0pt}{2.2ex}{1.5ex}
\titlespacing*{\subsection}{0pt}{2ex}{1.2ex}
\titlespacing*{\paragraph}{\parindent}{1ex}{0.6em}

\newcommand{\ve}{\varepsilon}

\newcommand{\ii}{\mathrm{i}}
\newcommand{\dd}{\mathrm{d}}
\newcommand{\R}{\mathbb{R}}
\newcommand{\norm}[1]{\left\lVert #1\right\rVert}
\newcommand{\ket}[1]{|#1\rangle}
\newcommand{\bra}[1]{\langle #1|}

\newcommand{\Ymat}{\mathsf Y}

\newcommand{\Dmat}{\mathsf D}

\newcommand{\qP}{q_{\rm P}}
\newcommand{\wtO}{\widetilde O}
\newcommand{\poly}{\operatorname{poly}}

\newenvironment{acknowledgments}{\par\medskip\noindent\textit{Acknowledgments.---}}{\par}
\theoremstyle{definition}
\newtheorem{theorem}{Theorem}
\theoremstyle{plain}
\newtheorem{lemma}{Lemma}

\title{Wavelength-Uniform Quantum Algorithms\\for Mixed-State Quantum Dynamics}
\author[1,2,3]{Shi Jin}
\author[4,5,6]{Chuwen Ma}
\affil[1]{School of Mathematical Sciences, Shanghai Jiao Tong University, Shanghai 200240, China}
\affil[2]{Institute of Natural Sciences, Shanghai Jiao Tong University, Shanghai 200240, China}
\affil[3]{MOE-LSC, Shanghai Jiao Tong University, Shanghai 200240, China}
\affil[4]{School of Mathematical Sciences, East China Normal University, Shanghai 200241, China}
\affil[5]{Key Laboratory of MEA, Ministry of Education, East China Normal University, Shanghai 200241, China}
\affil[6]{Shanghai Key Laboratory of PMMP, East China Normal University, Shanghai 200241, China}
\date{September 28, 2026}

\begin{document}
\twocolumn[\begin{@twocolumnfalse}
\maketitle
\vspace{-1em}
\begin{center}\small
\href{mailto:shijin-m@sjtu.edu.cn}{shijin-m@sjtu.edu.cn}\qquad
\href{mailto:cwma@math.ecnu.edu.cn}{cwma@math.ecnu.edu.cn}
\end{center}
\begin{abstract}
One of the main challenges in numerical  simulation of quantum dynamics is the prohibitive cost in the semi-classical regime, in which the de Broglie wave length is small compared with the characteristic length scale and the solution is highly oscillatory. For the von-Neumann equation for mixed-state quantum dynamics, this difficulty is overcome by using the Weyl variables,  under which the solution is not oscillatory. Furthermore, we use the integral representation of the potential difference, which robustly captures the classical limit as the semi-classical parameter approaches zero. By using exact Hermite moments and quantum singular value transformation to treat the polynomial and sparse coordinate matrices of the dense projected  smooth non-polynomial potentials respectively, we obtain  a  quantum algorithm efficient for {\it all} ranges of wave lengths, 
with complexity {\it polynomial} in the spatial dimension and discretization and query bounds containing {\it no} negative powers of the small wavelength.
Thus  it  can capture the correct physical observables even if the spatial grid does not resolve the frequency, hence defying the Nyquist-Shannon sampling theorem.
\end{abstract}
\vspace{1em}\end{@twocolumnfalse}]

\textbf{Introduction.---}
In many physical and chemical problems, both quantum and classical regimes co-exist, for example the 
Born--Oppenheimer molecular dynamics and nonadiabatic surface hopping
in photochemistry \cite{Born1927,worth2004,Tully1990,Schutte2002},  proton
transfer in enzymes and in solution \cite{Antoniou1997,Aqvist1993}, and
mixed quantum--classical dynamics and strong-field ionization
\cite{Kapral2006,Becker2012}. In the semi-classical regime, the solutions can be highly oscillatory, computationally resolving them  can be prohibitively expensive,  especially in high-dimensional problems. This raises a fundamental  question: 
{\it  Can a single quantum algorithm overcome the curse of dimensionality, remain uniformly accurate across the quantum and semiclassical regimes, and capture physical observables without resolving the rapid oscillations of the quantum state?}

For the mixed-state von Neumann equation, we develop a Weyl--Hermite quantum algorithm, which combines
the Weyl-variable formulation with suitable Hermite discretizations and sparse block encoding for the potential terms, 
to positively answer this question. 

Consider the dimensionless von Neumann equation
\begin{equation}\label{vN} \ii\ve\partial_t\widehat\rho^\ve =[\widehat H^\ve,\widehat\rho^\ve], \qquad \widehat H^\ve=-\frac{\ve^2}{2}\Delta+V(X), \end{equation} where \(V:\R^d\to\R\) is a real-valued scalar potential and \(\widehat\rho^\ve\) is a density operator with integral kernel \(\rho^\ve(t;X,Y)\), \(X,Y\in\R^d\). Here the commutator \([A,B]=AB-BA\), and \(\ve\) is the ratio of the reduced de Broglie wavelength to the characteristic length scale, or the square root of the mass ratio between electrons and nuclei in the Born-Oppenheimer approximation \cite{Born1927}. The semiclassical regime corresponds to \(\ve\ll1\). 

Numerical simulation of quantum dynamics faces two major challenges: high dimensionality and highly oscillatory semiclassical states. A wavelength-resolving computation with mesh size  \(h=O(\ve)\) stores
\(O(\ve^{-2d})\) density-matrix kernel samples on a fixed domain
\cite{BaoJinMarkowich2002,JMS2011,LasserLubich2020}:
on a unit box, \(\ve=10^{-2}\) and \(d=10\) already requires
\(10^{40}\) samples. Often $d$ can even be of $O(10^2)$ \cite{Wehrle2014Oligothiophenes}.  
A single-phase Wentzel--Kramers--Brillouin (WKB) ansatz  breaks
down beyond caustics  since the solutions become multi-valued \cite{Arnold1967,EngRun2003,SMM2003}. It is also not valid for $\ve=O(1)$. 
Gaussian wave-packet superpositions can pass through caustics,
but they require $h=O(\sqrt{\ve})$, and the methods are not valid if \(\ve=O(1)\)
\cite{Hagedorn1980,Heller1987,JMS2011,LasserLubich2020}.
The time or Trotter splitting methods allow
\(\ve\)-independent time steps for computing physical observables but still require $h=O(\ve)$
\cite{BaoJinMarkowich2002,GolseJinPaul2021,LasserLubich2020,FangQu2026}.
Quantum algorithms, including those for semiclassical dynamics,
encode grid-based wave functions using a number of data qubits
logarithmic in the grid size
\cite{Georgescu2014QuantumSimulation,Kassal2008ChemicalDynamics,JinLiLiu2022}, and  consequently amplitude encoding alone removes neither the wavelength-scale
grid requirement nor the inverse power of \(\ve\)-dependence of evolution costs.
Our starting point is the Weyl-variable formulation. Introduce \(x=(X+Y)/2\) and \(y=(X-Y)/\ve\), and set \(R^\ve(t,x,y)=\rho^\ve(t;X,Y)\). Equation~\eqref{vN} becomes \begin{align} \partial_tR^\ve &=\ii\sum_{j=1}^d\partial_{x_j}\partial_{y_j}R^\ve -\ii U_\ve(x,y)R^\ve,\nonumber\\ U_\ve(x,y) &=\frac{V(x+\ve y/2)-V(x-\ve y/2)}{\ve}. \label{eq:weylpde} \end{align} 
For suitable mixed-state families, this formulation admits
\(\ve\)-uniform regularity, as proved by Filbet and Golse
\cite{FilbetGolse2025,FilbetGolse2026}, based on which they   constructed a classical
Hermite method with \(\ve\)-uniform approximation estimates.
They used a truncated Taylor expansion which  works efficiently for $\ve \ll1$, but is not accurate for $\ve=O(1)$. 

A numerical difficulty in evaluating $U_\ve$ is that when
$\ve\ll1$, due to finite-precision arithmetic, the computer
may not accurately distinguish $V(x+\ve y/2)$ from
$V(x-\ve y/2)$.
Moreover, combining separate block-encodings of these shifted
potentials to encode $U_\ve$ introduces an $\ve^{-1}$
normalization factor.
To address these issues, 
we decompose  \(V=V_{\rm P}+V_{\rm F}\), where \(V_{\rm P}\) has fixed polynomial degree and \(V_{\rm F}\) is a smooth non-polynomial that admits a Fourier integral or periodic series. 
For $V_{\rm P}$, the corresponding potential difference
$U_\ve^{\rm P}$ has an exact finite expansion in nonnegative
even powers of $\ve$.
We encode it using exact projected Hermite moments, whose
matrices remain {\it sparse} for fixed polynomial degree.
For $V_F$, our new idea starts from  using the integral formulation and the  Fourier transform
\(V_{\rm F}(x)=(2\pi)^{-d}\int\widehat V_{\rm F}(\xi)
e^{\ii\xi\cdot x}\,\dd\xi\), we obtain 
\begin{align}
    U_\ve^{\rm F}(x,y)
&=\int_{-1/2}^{1/2}
 y\cdot\nabla V_F(x+s\ve y)\,\dd s
 \notag\\
 &=\frac{\ii}{(2\pi)^d}\int_{\R^d}
    \widehat V_{\rm F}(\xi)e^{\ii\xi\cdot x}(\xi\cdot y)
    \operatorname{sinc}\!\left(\frac{\ve}{2}\xi\cdot y\right)
    \dd\xi \notag\\
    &\xrightarrow{\ve\to0}y\cdot\nabla V_F(x),
    \label{eq:fourier-U}
\end{align}
where \(\operatorname{sinc}z=\sin z/z\) and
\(\operatorname{sinc}0=1\).
It is clear that this formulation is most suitable for the $\ve \to 0$ limit as shown in the last line of the above equation, since it naturally leads to the forcing term in the classical Liouville equation which is the classical limit of the von Neumann equation via the Wigner transform. In addition,  although this will lead to a dense Hermite matrix, it is important to notice that for each Fourier mode, multiplication by \(\xi\cdot y\) has a {\it sparse}
Hermite matrix. We  apply the  quantum singular value transformation (QSVT) to this sparse matrix in an enlarged
Hermite space, then project onto the retained modes and combine the
Fourier contributions using the linear combination of
unitaries (LCU) technique \cite{LowChuang2017,Gilyen2019}. 

Under the stated regularity and structured-access assumptions,
our quantum algorithms possess the following properties:
\textit{i)} A single discretization giving accurate physical  observables
with cost {\it independent} of   \(0<\ve\le1\), with end-to-end resource bounds
containing {\it no} inverse powers of \(\ve\).
\textit{ii)} A QSVT--LCU block encoding of the 
potential to prescribed accuracy, {\it without} dense matrix assembly or
semiclassical Taylor truncation.
\textit{iii)} Query and gate costs polynomial in \(d\) for fixed
time and accuracy, under additional dimension-stable assumptions
including efficient input and observable preparation, avoiding
the exponential cost of a direct classical full-tensor
discretization for physical observables.
To our knowledge, this is the first end-to-end quantum algorithm
achieving {\it $\ve$-independent}  cost for the mixed-state quantum dynamics. As a result,  
we compute accurately the physical observables without the  Nyquist--Shannon resolution--which requires several grid points per wave length-- of the original
quantum state.

\textbf{Discretization.---}
We restrict and periodize the problem in \(x\) on a box
\(\Omega_{x,\delta}\simeq\mathbb T^d\), with observable error
\(O(\delta)\) independent of  \(\ve\) under the stated domain-control
assumptions \cite{SuppMat}. We retain \(K\) Hermite modes per
\(y\) direction and use \(M\) grid points per spatial direction.

\emph{Hermite discretization in \(y\).---}
Let \(I_K=\{0,\ldots,K-1\}^d\) and
\(\Phi_{\mathbf k}(y)=\prod_{j=1}^d\Phi_{k_j}(y_j)\), where
\(\Phi_k\) are normalized Hermite functions. The expansion
\(R_K^\ve=\sum_{\mathbf k\in I_K}R_{\mathbf k}(t,x)
\Phi_{\mathbf k}(y)\) and the Hermite ladder identity
\(\partial_{y_j}\Phi_{\mathbf k}
=\sqrt{k_j/2}\,\Phi_{\mathbf k-\mathbf e_j}
-\sqrt{(k_j+1)/2}\,\Phi_{\mathbf k+\mathbf e_j}\)
\cite{Tang1993} give
\begin{align}
    \ii\partial_tR_{\mathbf k}
    &=\sum_{j=1}^d\left(
    \sqrt{\frac{k_j}{2}}\,\partial_{x_j}R_{\mathbf k-\mathbf e_j}
    -\sqrt{\frac{k_j+1}{2}}\,\partial_{x_j}R_{\mathbf k+\mathbf e_j}
    \right)\nonumber\\
    &\quad+\sum_{\mathbf l\in I_K}
    [U_{\ve,K}(x)]_{\mathbf k\mathbf l}R_{\mathbf l},
    \label{eq:Hermite-system-d}
\end{align}
where \(\mathbf e_j\) is the \(j\)th coordinate vector,
\(R_{\mathbf m}=0\) for \(\mathbf m\notin I_K\), and
\([U_{\ve,K}(x)]_{\mathbf k\mathbf l}
=\int_{\R^d}U_\ve(x,y)\Phi_{\mathbf k}(y)
\Phi_{\mathbf l}(y)\,\dd y\).
The same ladder identity defines the skew-Hermitian matrix
\(\mathsf D_{y_j,K}\) representing \(\partial_{y_j}\), with only
\((\mathsf D_{y_j,K})_{\mathbf k,\mathbf k+\mathbf e_j}
=-(\mathsf D_{y_j,K})_{\mathbf k+\mathbf e_j,\mathbf k}
=\sqrt{(k_j+1)/2}\)
as nonzero entries.
The initial coefficients are the Hermite projections of \(R^\ve(0)\).

\emph{Spatial discretization in \(x\).---}
Set \(I_M=\{0,\ldots,M-1\}^d\). On the periodic mesh of size \(h\),
approximate \(\partial_{x_j}\) by an order-\(r_x\), fixed-width,
skew-Hermitian finite-difference matrix \(\mathsf D_{x_j,h}\).
Collecting the coefficients into
\(\mathbf R(t)=h^{d/2}\sum_{\mathbf i\in I_M,\mathbf k\in I_K}
R_{\mathbf k}(t,x_{\mathbf i})\ket{\mathbf i}\ket{\mathbf k}\)
gives
\begin{equation}
    \ii\partial_t\mathbf R
    =(H_{\rm tr}+U_{\ve,h,K}^{\rm ex})\mathbf R,
    \label{eq:semidiscrete-d}
\end{equation}
where \(H_{\rm tr}=-\sum_j\mathsf D_{x_j,h}\otimes\mathsf D_{y_j,K}\) corresponds to the spatial derivative (kinetic energy) term,
and \(U_{\ve,h,K}^{\rm ex}=\sum_{\mathbf i\in I_M}
\ket{\mathbf i}\bra{\mathbf i}\otimes U_{\ve,K}(x_{\mathbf i})\) corresponds to the potential term.
Each derivative matrix acts only on its corresponding \(j\)th index.
The generator is Hermitian for real \(V\).

\emph{Remark.---}
Alternative discretizations in \(x\) include centered finite-volume
and Fourier spectral methods, as well as Hermite expansions on
\(\R^d\) under suitable decay assumptions.

\textbf{Quantum implementation.---}
A block encoding embeds a normalized matrix into a unitary
\cite{Gilyen2019}. We block-encode an approximation of the generator
in Eq.~\eqref{eq:semidiscrete-d}, treating its sparse and dense
components separately.

\emph{Sparse terms.---}
The kinetic operator \(H_{\rm tr}\) admits a sparse block encoding. For the polynomial component
\(V_{\rm P}\) of degree \(D_{\rm P}\), the symmetric difference is 
\begin{equation}
    U_\ve^{\rm P}(x,y)
    =\sum_{m=0}^{\qP}\frac{\ve^{2m}}{2^{2m}}
    \sum_{|\nu|=2m+1}
    \frac{(\partial^\nu V_{\rm P})(x)y^\nu}{\nu!},
    \label{eq:poly-U}
\end{equation}
where \(\qP=\lfloor(D_{\rm P}-1)/2\rfloor\) and \(\nu\) is a
multi-index. This finite expansion contains {\it only nonnegative} powers
of \(\ve\), with no asymptotic remainder.
We encode its monomials using tensor products of the exact
\(K\times K\) Hermite moment matrices
\((\Ymat_K^{[r]})_{\ell k}
=\int_{\R}y^r\Phi_\ell(y)\Phi_k(y)\,\dd y\).
In general, \(\Ymat_K^{[r]}\ne(\Ymat_K^{[1]})^r\): powers of a
truncated coordinate matrix do not reproduce the exact projection.
This construction gives the exact projected block
\(U_{\ve,h,K}^{\rm P}\), with row sparsity polynomial in \(d\)
for fixed \(D_{\rm P}\) \cite{SuppMat}.

\emph{Smooth non-polynomial potential.---}Choose a paired quadrature
\(V_{\rm F}(x)\approx\sum_{q\in\mathcal Q}a_qe^{\ii\xi_q\cdot x}\),
where \(a_q=w_q\widehat V_{\rm F}(\xi_q)/(2\pi)^d\),
\(w_{-q}=w_q>0\), and \(\xi_{-q}=-\xi_q\).
Thus \(a_{-q}=\overline{a_q}\) for real \(V_{\rm F}\).
For periodic potentials, \(a_q\) are the truncated Fourier-series
coefficients. The zero mode is omitted because it cancels in the difference.
Truncating the coordinate matrices to \(K\) modes before evaluating the potential matrix function generally does not reproduce the exact \(K\)-mode projection of the potential. To control this discrepancy, we enlarge each Hermite register to \(\widetilde K=K+L_{\rm buf}\) modes and evaluate the matrix function in this larger space before projecting back onto the retained modes. Thus, \(\widetilde K\) controls the accuracy of the potential approximation, while \(K\) remains the solution cutoff.
Define the sparse ladder matrix
\(\Ymat_N=\sum_{k=0}^{N-2}\sqrt{(k+1)/2}
(\ket{k+1}\bra{k}+\ket{k}\bra{k+1})\), and let
\(\Ymat_{j,N}\) act on the \(j\)th Hermite factor.
Let \(J_d\) embed the retained \(K^d\)-dimensional space into the
enlarged \(\widetilde K^d\)-dimensional space by zero padding, and set
\(\mathsf A_{q,\widetilde K}=\sum_j\xi_{q,j}\Ymat_{j,\widetilde K}\)
and \(\mathsf E_q\ket{\mathbf i}
=e^{\ii\xi_q\cdot x_{\mathbf i}}\ket{\mathbf i}\).
Equation~\eqref{eq:fourier-U} then leads to the finite-section approximation
\begin{align}
	\mathsf F_{q,K}^{\rm fs}
	&:=J_d^\dagger\mathsf A_{q,\widetilde K}
	\operatorname{sinc}\!\left(\frac{\ve}{2}\mathsf A_{q,\widetilde K}\right)J_d,
	\nonumber\\
	U_{\ve,h,K}^{\rm F,fs}
	&:=\ii\sum_{q\in\mathcal Q}a_q\mathsf E_q\otimes\mathsf F_{q,K}^{\rm fs}.
	\label{eq:Fourier-finite-d}
\end{align}
The buffer $L_{\rm buf}$ controls the error from replacing the infinite Hermite
operator by a finite section \cite{SuppMat}.

We use QSVT to implement a polynomial approximation to the matrix function
through queries to a block encoding of the {\it sparse} matrix
\(\mathsf A_{q,\widetilde K}\). The number of queries is proportional
to the polynomial degree.
Projection and LCU then produce the generally dense potential block.
It is important to notice that its implementation cost is controlled by the normalization and
polynomial degree below, without evaluating or storing dense matrix
entries.

The sparse matrix \(\mathsf A_{q,\widetilde K}\) has block-encoding
normalization \(\beta_q=\sqrt{2\widetilde K}\norm{\xi_q}_1\),
where \(\norm{\xi_q}_1=\sum_{j=1}^d|\xi_{q,j}|\).
Set \(\mathcal V_{1,\mathcal Q}=\sum_q|a_q|\norm{\xi_q}_1\),
\(\beta_{\max}=\max_q\beta_q\), and
\(\alpha_{\rm pot}^{\rm F}=\sqrt{2\widetilde K}\mathcal V_{1,\mathcal Q}\).
QSVT, followed by compression and LCU, yields a Hermitian
approximation \(U_{\ve,h,K}^{\rm F,QSVT}\) to
\(U_{\ve,h,K}^{\rm F,fs}\) in Eq.~\eqref{eq:Fourier-finite-d},
with operator error at most \(\delta_{\rm pot}\).
The QSVT polynomials have maximum degree
\begin{equation}
    n_{\rm F}
    =O\!\left(1+\ve\beta_{\max}
    +\log\!\left[1+
    \frac{\alpha_{\rm pot}^{\rm F}}{\delta_{\rm pot}}\right]\right).
    \label{eq:Fourier-resources-main}
\end{equation}
Thus neither the normalization nor the degree introduces an inverse
power of \(\ve\) at fixed discretization. The approximation retains
the full finite-\(\ve\) matrix function; Fourier quadrature and
finite-section errors are controlled separately.

\emph{Hamiltonian simulation and resources.---}
The resulting block encoding represents
the Hermitian matrix
\begin{equation}
    H_{\rm WH}=H_{\rm tr}+U_{\ve,h,K}^{\rm P}
    +U_{\ve,h,K}^{\rm F,QSVT}.
    \label{eq:H-WH-main}
\end{equation}
On \(\Omega_x=\Omega_{x,\delta}\), set
\(\Gamma_{m,d}^{\rm P}(\Omega_x)=\sum\limits_{|\nu|=2m+1}
\frac{\norm{\partial^\nu V_{\rm P}}_{L^\infty(\Omega_x)}}{\nu!}\).
Its normalization satisfies
\begin{align}
    \alpha_{\rm WH}
    =O_{r_x,D_{\rm P}}\!\Big(&dh^{-1}K^{1/2}
    +\sum_{m=0}^{\qP}\ve^{2m}K^{m+1/2}
    \Gamma_{m,d}^{\rm P}(\Omega_x)\nonumber\\
    &+\sqrt{\widetilde K}\,\mathcal V_{1,\mathcal Q}\Big).
    \label{eq:alpha-fixed-main}
\end{align}
We assume reversible evaluation of polynomial derivatives, Fourier
coefficients, frequencies, and grid phases, coherent Fourier-weight
preparation, and access to QSVT phases controlled by the Fourier
index, with costs {\it uniform} in \(\ve\) at the allocated precision.
Outer qubitization and quantum signal processing implement
\(e^{-\ii TH_{\rm WH}}\) \cite{LowChuang2019}; their cost includes
the inner QSVT work required by each Fourier-block query.

\textbf{Inputs and physical observables.---}
Uniform simulation cost yields an $\ve$-uniform cost for the physical observables 
when input preparation, readout, and their normalization factors are
also controlled. Let \(\mathbf R_0\) contain the sampled Hermite
coefficients of \(R^\ve(0)\), and define
\begin{align}
	(\mathbf b_a)_{\mathbf i,\mathbf k}
	&=h^{d/2}\int_{\R^d}\check a(x_{\mathbf i},y)
	\Phi_{\mathbf k}(y)\,\dd y,\nonumber\\
	A_{\rm WH}(T)
	&=\langle\mathbf b_a,e^{-\ii TH_{\rm WH}}\mathbf R_0\rangle,
	\label{eq:observable}
\end{align}
where \(a(x,p)\) is a real phase-space observable and
$
\check a(x,y)
=\frac{1}{(2\pi)^d}\int_{\R^d}
a(x,p)e^{\ii p\cdot y}\,\dd p
$
is its inverse Fourier transform in momentum.
For example, the choices
\(a(x,p)=w(x,p)\),
\(a(x,p)=p_jw(x,p)\), and
\(a(x,p)=\tfrac12|p|^2w(x,p)\)
give windowed versions of the mass, momentum, and kinetic-energy
observables.
Here \(w(x,p)\) is a phase-space window selecting the spatial
region and momentum range of interest.
For known norms \(s_R=\norm{\mathbf R_0}_2\) and
\(s_a=\norm{\mathbf b_a}_2\), overlap estimation on normalized
vectors gives \(A_{\rm WH}(T)/(s_as_R)\)
\cite{Rall2020,Brassard2002}.
We assume \(s_R,s_a,\norm{\check a}_{L^2_{x,y}}=O(1)\), uniformly
in \(\ve\) and \(d\), and efficient preparation of the normalized
input and readout states. 

For example, a trace-one Gaussian mixed state with Weyl kernel
proportional to \(e^{-|x-x_0|^2/(2\sigma^2)}e^{-|y|^2/2}\),
\(\sigma>1/2\), has fixed spatial and momentum widths and physical
coherence length \(O(\ve)\). Its Weyl profile is independent of \(\ve\),
with \(\norm{R_0^\ve}_{L^2}^2=(2\sigma)^{-d}\le1\).
Only \(\mathbf k=\mathbf0\) is occupied initially, and its factorized
Gaussian amplitudes admit preparation with cost polynomial in
\(d,\log M,\log K,\log\delta^{-1}\) \cite{GroverRudolph2002};
Gaussian readout states admit a similar construction.
The restriction on inputs is substantive:
\(\norm{R_0^\ve}_{L^2}^2=\ve^{-d}
\operatorname{Tr}[(\widehat\rho_0^\ve)^2]\)
\cite{FilbetGolse2026}. A resolved trace-one pure-state family has
\(s_R=\Theta(\ve^{-d/2})\) and lies outside the uniform end-to-end
guarantee, although the Hamiltonian block-encoding bounds remain valid.

\textbf{Uniform accuracy and main result.---}
We first fix \(d\) and \(T\).
For the real phase-space observable \(a\) introduced above,
define its exact expectation at time \(T\) by
$A^\ve(T):=\langle \check a,R^\ve(T)\rangle_{L^2(\mathbb R^{2d})}$.
Its Weyl--Hermite approximation is \(A_{\rm WH}(T)\)
defined above.
We use \(\ve\)-uniform weighted regularity in Weyl variables,
following the regularity propagation analysis of Filbet and Golse
\cite{FilbetGolse2026}.
The regularity conditions and discretization error estimates
are detailed in Ref.~\cite{SuppMat}.
For \(p\ge1\) and \(2p\ge r_x\), with sufficient regularity,
the Hermite error is \(O(K^{-(p-1/2)})\).
We choose \(K\) to make this error \(O(\delta)\), and
then choose \(h\) so that the \(O(h^{r_x})\) spatial
error at this cutoff is also \(O(\delta)\).

For the  contribution of potential $V_F$, the frequency cutoff
and quadrature are selected to control the
projected-potential error, using the first absolute
Fourier-moment tail for truncation.
The buffer \(L_{\rm buf}\) controls the finite-section
error in Eq.~\eqref{eq:Fourier-finite-d}.
With the corresponding error budgets, all discretization
parameters can be chosen independently of \(\ve\) to yield
\begin{equation}
	\sup_{0<\ve\le1}
	|A_{\rm WH}(T)-A^\ve(T)|=O(\delta).
	\label{eq:uniform-consistency-main}
\end{equation}

To control how these choices grow with \(d\), we
additionally require polynomial growth of 
\(\Gamma_{m,d}^{\rm P}(\Omega_{x,\delta})\),
the zeroth and first Fourier moments and their
quadrature counterparts, the required Fourier bandwidth,
and the access and preparation costs, at fixed
\(T,\delta,p,r_x,D_{\rm P}\) \cite{SuppMat}.
Then the Hermite cutoff
can be chosen as
\begin{equation}
	K_{d,\delta}
	=O\!\left(
	\operatorname{poly}(d)\,\delta^{-2/(2p-1)}
	\right).
	\label{eq:K-rate-main}
\end{equation}
The required inverse mesh size \(h^{-1}\) and enlarged
Hermite cutoff \(\widetilde K=K+L_{\rm buf}\) likewise
grow at most polynomially in \(d\).

\begin{theorem}[Uniform observable complexity]
	\label{thm:resource-main}
	Under the assumptions stated above and detailed in
	Ref.~\cite{SuppMat}, choose the discretization parameters
	to satisfy Eq.~\eqref{eq:uniform-consistency-main}.
	Let \(\alpha_{{\rm WH},d}(\ve;\delta)\) be the block-encoding
	normalization satisfying Eq.~\eqref{eq:alpha-fixed-main}
	at these parameters on \(\Omega_{x,\delta}\), and set
	\(\Lambda_{d,\delta}=\max\{1,s_Rs_a/\delta\}\).
	For \(s_Rs_a>0\), assume preparation oracles
	\(O_R\ket0=\mathbf R_0/s_R\) and
	\(O_a\ket0=\mathbf b_a/s_a\), with costs uniform in \(\ve\).
	
	An estimate \(\widetilde A(T)\) satisfying
	\begin{equation}
		|\widetilde A(T)-A^\ve(T)|=O(\delta),
		\qquad 0<\ve\le1,
		\label{eq:end-to-end-error-main}
	\end{equation}
	can be obtained with constant success probability using
	\begin{equation}
		N_H=\wtO\!\left[
		\Lambda_{d,\delta}
		\left\{
		T\alpha_{{\rm WH},d}(\ve;\delta)
		+\log(1+\Lambda_{d,\delta})
		\right\}
		\right]
		\label{eq:queries}
	\end{equation}
	queries to the block encoding of \(H_{\rm WH}\), together
	with \(O(\Lambda_{d,\delta})\) calls to each preparation
	oracle and its inverse. Here \(\wtO\) is $O$ by suppressing the  logarithmic factors.
	
	For non-empty set \(\mathcal Q\), 
	Hamiltonian block-encoding queries uses
	\[
	N_{\rm ladder}=O(n_{\rm F}N_H)
	\]
	queries to the Hermite ladder block encodings in the contribution of $V_F$, where \(n_{\rm F}\) is the maximum
	inner QSVT degree in Eq.~\eqref{eq:Fourier-resources-main}.
	If \(\mathcal Q\) is an empty set,
    the $V_F$ contribution
	is omitted and \(N_{\rm ladder}=0\).

    With the discretization parameters chosen independently of
\(\ve\), both query bounds are uniform for \(0<\ve\le1\).
	Under the additional dimension assumptions stated above,
	they are polynomial in \(d\) for  fixed
	\(T,\delta,p,r_x,D_{\rm P}\).
	With the stated bounds on access and preparation costs,
	gate complexity has the same polynomial dependence on
	\(d\) and contains {\it no inverse powers}  of \(\ve\).
	If \(s_Rs_a=0\), return \(\widetilde A(T)=0\).
\end{theorem}

The estimate is
\(\widetilde A(T)=s_as_R\widetilde\mu\), with simulation,
normalized-state preparation, and overlap errors each
chosen as \(O(\Lambda_{d,\delta}^{-1})\).
Equations~\eqref{eq:Fourier-resources-main} and
\eqref{eq:alpha-fixed-main}, evaluated with discretization
parameters chosen independently of \(\ve\), give upper bounds
on \(n_{\rm F}\) and \(\alpha_{\rm WH}\) that are uniform
for \(0<\ve\le1\).
Circuit depth, auxiliary workspace, and state-preparation
costs are detailed in Ref.~\cite{SuppMat}.

\begin{figure*}[t!]
 \centering
 \includegraphics[width=\textwidth]{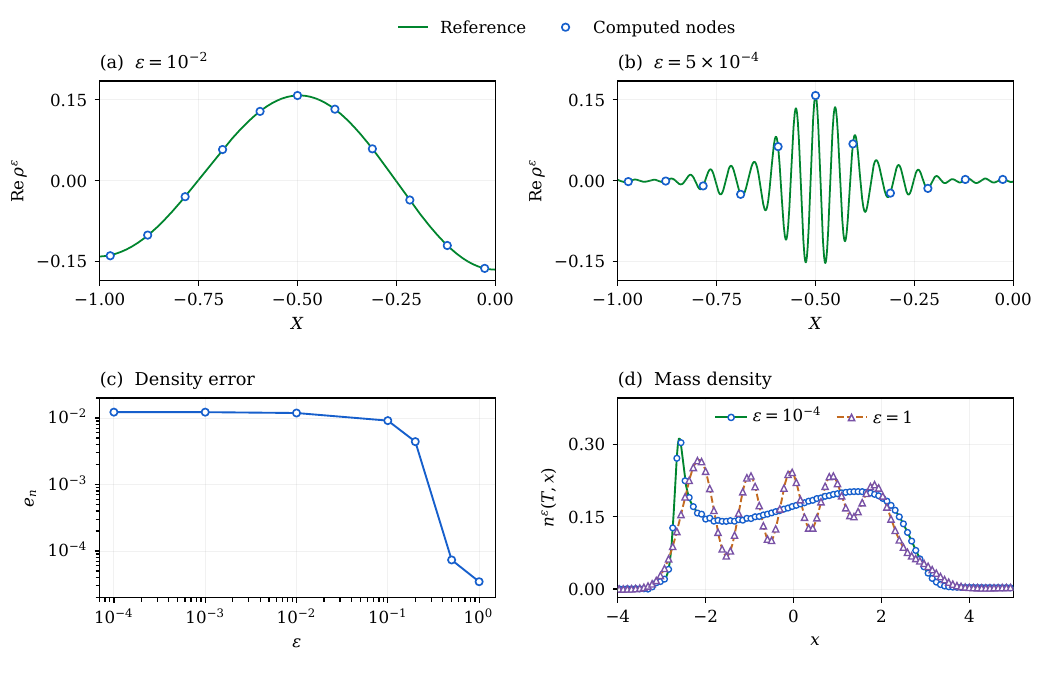}
 \caption{Morse-potential dynamics at \(T=19\) with fixed
 \(M=256\), \(K=192\), and \(h=0.09375\).
 (a),(b) Non-diagonal kernel cuts \(g^\ve(X)\) for
 \(\ve=10^{-2}\) and \(5\times10^{-4}\).
 Green curves are the independent reference; blue circles are
 all eleven computed nodes in the displayed interval, with
 \(X_i=(2x_i+0.01)/1.98\) and spacing \(2h/1.98\).
 (c) Relative maximum nodal density error \(e_n\) versus \(\ve\),
 normalized by the peak reference density over the full computational
 domain. 
 (d) Mass densities for \(\ve=10^{-4}\) and \(1\) from the same
 initial Weyl profile; curves show references and markers show
 all 96 computed nodes in the displayed interval for each \(\ve\).
 Every Weyl--Hermite run uses centered differences in \(x\), fourth-order
 Yoshida splitting with \(\Delta t=0.01\), and \(\widetilde K=224\).}
 \label{fig:independent-validation}
\end{figure*}

\textbf{A numerical benchmark.---}
We benchmark the Weyl--Hermite discretization with the Morse potential
and a common initial Weyl profile,
\begin{align}
 V(x)&=20(1-e^{-0.16x})^2,\qquad \sigma=0.6,\nonumber\\
 R_0(x,y)&=\frac{1}{\sqrt{2\pi}\sigma}
 e^{-(x-4)^2/(2\sigma^2)-y^2/2}.
 \label{eq:morse-benchmark}
\end{align}
All Weyl--Hermite runs retain the full finite-\(\ve\) potential difference and use
\(M=256\), \(K=192\), \(h=0.09375\), and \(T=19\), with
centered differences in \(x\) and fourth-order time splitting
at \(\Delta t=0.01\).
To apply the smooth non-polynomial potential construction,
we first extend the Morse potential to a periodic function.
Its periodic Fourier representation and implementation details are given
in Ref.~\cite{SuppMat}.

Figure~\ref{fig:independent-validation}(a),(b) shows
\(g^\ve(X)=\operatorname{Re}\rho^\ve(T;X,0.98X-0.01)\)
for \(\ve=10^{-2}\) and \(5\times10^{-4}\), respectively.
In panel~(b), these nodes do not resolve the oscillations,
yet the kernel values remain accurate. This evaluation exploits the
Weyl representation and its known coordinate map.
Panel~(c) reports errors in the mass density
\(n^\ve(T,x)=R^\ve(T,x,0)\) for seven wavelengths at the fixed
discretization. The error approaches a \(1.225\times10^{-2}\) plateau as
\(\ve\) decreases, consistent with \(\ve\)-uniform approximation
estimates. Panel~(d) shows substantially different densities at
\(\ve=10^{-4}\) and \(1\), despite identical initial Weyl profiles:
the small-\(\ve\) result cannot replace the finite-wavelength dynamics
at \(\ve=1\).

Independent reference solutions use a fine Fourier discretization in
both Weyl variables. Further numerical examples and implementation
details are provided in Ref.~\cite{SuppMat}.
These classical computations test the discretization. Quantum
implementation errors and resource bounds are controlled separately
by the estimates above.

\textbf{Summary.---}
The Weyl--Hermite algorithm takes  advantage of uniform regularity in Weyl
variables, with a quantum encoding of the full potential difference.
Exact projected Hermite moments treat the polynomial contribution,
while QSVT in an enlarged Hermite basis, followed by projection and
LCU, encodes generally dense Fourier potential couplings through
sparse coordinate matrices. This construction avoids dense-matrix
assembly, semiclassical Taylor truncation as done in \cite{FilbetGolse2025}, and inverse-\(\ve\)
normalization. Under the assumptions of
Theorem~\ref{thm:resource-main}, we achieve  \(\ve\)-independent
computational cost for physical observables, %
including simulation and readout errors,
with end-to-end resource bounds containing {\it no inverse powers}
of \(\ve\). Under additional dimension-stable assumptions, query
and gate costs are polynomial in the spatial dimension \(d\) for fixed time and accuracy,
with input preparation and observable normalization included. In particular, 
for the stated problem class, the algorithm obtains accurate physical 
observables without resolving numerically the possibly very small wave length.

\begin{acknowledgments}
SJ acknowledges support from NSFC Grant No.~12341104, the Shanghai Pilot
Program for Basic Research, the STCSM Grant No.~24LZ1401200, the Shanghai
Jiao Tong University 2030 Initiative, and the Fundamental Research Funds
for the Central Universities. CM was supported by the Fundamental and
Interdisciplinary Disciplines Breakthrough Plan of the Ministry of
Education of China (JYB2025XDXM112) and by NSFC Grant No.~12501607.
CM was also partially supported by STCSM Grant No.~22DZ2229014.
\end{acknowledgments}

\paragraph*{Data availability.---}
Data and source code for Fig.~\ref{fig:independent-validation} accompany
the Supplemental Material.

% ================================================================
% Supplemental Material: new page, single column, independent numbering.
% ================================================================
\clearpage
\onecolumn
\allowdisplaybreaks[2]
\renewcommand{\epsilon}{\varepsilon}
\setcounter{section}{0}
\setcounter{subsection}{0}
\setcounter{equation}{0}
\setcounter{figure}{0}
\setcounter{table}{0}
\setcounter{lemma}{0}
\renewcommand{\theequation}{S\arabic{equation}}
\renewcommand{\thefigure}{S\arabic{figure}}
\renewcommand{\thetable}{S\arabic{table}}
\renewcommand{\thelemma}{S\arabic{lemma}}
\renewcommand{\thesection}{\Roman{section}}
\renewcommand{\thesubsection}{\Alph{subsection}}
\setcounter{secnumdepth}{2}
\makeatletter
\renewcommand{\p@subsection}{\thesection.}
\makeatother
% Unique PDF destinations after the counters are reset.
\renewcommand{\theHequation}{supp.\arabic{equation}}
\renewcommand{\theHfigure}{supp.\arabic{figure}}
\renewcommand{\theHtable}{supp.\arabic{table}}
\renewcommand{\theHsection}{supp.\arabic{section}}
\renewcommand{\theHsubsection}{supp.\arabic{section}.\arabic{subsection}}
\providecommand{\theHlemma}{\arabic{lemma}}
\renewcommand{\theHlemma}{supp.\arabic{lemma}}
\phantomsection
\label{supplement:start}
\pdfbookmark[0]{Supplemental Material}{supplement}

\begin{center}
\large\bfseries
Supplemental Material for ``Wavelength-Uniform Quantum Algorithms for
Mixed-State Quantum Dynamics''\par
\vspace{1.2em}
\normalsize\normalfont Shi Jin\textsuperscript{1,2,3} and Chuwen Ma\textsuperscript{4,5,6}\par
\vspace{0.6em}
\small\itshape
\textsuperscript{1}School of Mathematical Sciences, Shanghai Jiao Tong University, Shanghai 200240, China\par
\textsuperscript{2}Institute of Natural Sciences, Shanghai Jiao Tong University, Shanghai 200240, China\par
\textsuperscript{3}MOE-LSC, Shanghai Jiao Tong University, Shanghai 200240, China\par
\textsuperscript{4}School of Mathematical Sciences, East China Normal University, Shanghai 200241, China\par
\textsuperscript{5}Key Laboratory of MEA, Ministry of Education, East China Normal University, Shanghai 200241, China\par
\textsuperscript{6}Shanghai Key Laboratory of PMMP, East China Normal University, Shanghai 200241, China\par
\normalfont (Dated: September 27, 2026)
\end{center}
\vspace{0.5em}

\noindent\textit{Organization.---}
Section~\ref{sec:finite} fixes the finite Weyl--Hermite system and its
Gauss--Hermite representation. Sections~\ref{sec:proof1} and
\ref{sec:proof2} state and prove, respectively,
Lemma~\ref{lem:fixed-resources} and Lemma~\ref{lem:uniform-selection}.
Section~\ref{sec:proof3} proves the main theorem of the Letter.
Section~\ref{sec:io} discusses input and readout, and
Sec.~\ref{sec:numerics} presents three numerical experiments on Weyl profiles,
physical densities, and a specified terminal-pulse propagation.
Equations, figures, and lemmas in this Supplemental Material carry the
prefix S. Cited references are listed at the end of this document.
We keep explicit all dependence on \(d\) needed by the resource
bounds. The polynomial-in-\(d\) conclusion is conditional on the
dimension-stable regularity, Fourier-access, and input--readout assumptions
stated in Secs.~\ref{sec:proof1}--\ref{sec:io}; no dimension-uniform PDE
regularity theorem is asserted.

Throughout, \(\alpha,\beta,\gamma,\zeta\in\mathbb N_0^d\) are
multi-indices, with \(|\alpha|=\sum_j\alpha_j\),
\(x^\alpha=\prod_jx_j^{\alpha_j}\), and
\(\partial_x^\beta=\prod_j\partial_{x_j}^{\beta_j}\), and analogously
in \(y\). Dimension-dependent constants carry a subscript \(d\) and
are kept explicit in the final resource bounds.

The shared notation agrees with the Letter: \(M\) is the number of spatial
grid points per direction, \(K\) is the retained Hermite size,
\(\widetilde K=K+L_{\rm buf}\) is the auxiliary Hermite size,
\(\mathcal Q\) is the paired Fourier-index set, and
\(N_{\rm ladder}\) counts queries to the Hermite coordinate block encodings.
In the one-dimensional numerical experiments, \(Q\) counts positive
Fourier modes, so \(|\mathcal Q|=2Q\) when all retained coefficients are
nonzero. The theoretical basis \(\Phi_k\) is unscaled; numerical uses of
\(\Phi_k^{(\ell)}\) and \(\mathsf Y_N^{(\ell)}=\ell\mathsf Y_N\)
are stated explicitly.

\section{Finite Weyl--Hermite system}\label{sec:finite}

\subsection{Weyl variables and Hermite projection}

We use the Weyl variables of the Letter:
\begin{equation}
 x=\frac{X+Y}{2},\qquad y=\frac{X-Y}{\ve},\qquad
 R^\ve(t,x,y)=\rho^\ve(t;x+\ve y/2,x-\ve y/2).
 \label{eq:change}
\end{equation}
The evolution equation and the symmetric potential difference are
\begin{align}
 \partial_tR^\ve
 &=\ii\sum_{j=1}^d\partial_{x_j}\partial_{y_j}R^\ve
 -\ii U_\ve(x,y)R^\ve,\qquad
 U_\ve(x,y)
 =\frac{V(x+\ve y/2)-V(x-\ve y/2)}{\ve}
 =\int_{-1/2}^{1/2}y\cdot\nabla V(x+s\ve y)\,\dd s.
 \label{eq:weyl-supp}
\end{align}
The Weyl formulation and its one-dimensional Hermite discretization
were studied in Ref.~\cite{sm:FilbetGolse2025}. The norm and observable
pairings satisfy
\begin{equation}
 \norm{R^\ve}_2^2=\ve^{-d}\norm{\widehat\rho^\ve}_{\rm HS}^2,
 \qquad
 \langle\check a,R^\ve\rangle
 =\operatorname{Tr}[\widehat\rho^\ve\operatorname{Op}^{\rm W}_\ve(a)],
 \label{eq:pairing}
\end{equation}
where \(\check a(x,y)=(2\pi)^{-d}\int_{\R^d}a(x,p)e^{\ii y\cdot p}\,\dd p\)
and \(\operatorname{Op}^{\rm W}_\ve\) denotes Weyl quantization with
semiclassical parameter \(\ve\). The inner product is conjugate-linear
in its first argument. Equation~\eqref{eq:pairing} holds for real
\(a\); complex symbols are treated componentwise.

The normalized Hermite functions are
\(\Phi_k(y)=H_k(y)e^{-y^2/2}/[\pi^{1/4}\sqrt{2^kk!}]\) and
\(\Phi_{\mathbf k}(y)=\prod_{j=1}^d\Phi_{k_j}(y_j)\), where
\(H_k\) is the degree-\(k\) Hermite polynomial orthogonal with
weight \(e^{-y^2}\) and leading coefficient \(2^k\). With
\(\Phi_{-1}=0\), the ladder identities
\begin{equation}
 y\Phi_k=\sqrt{\frac{k+1}{2}}\Phi_{k+1}+\sqrt{\frac{k}{2}}\Phi_{k-1},
 \qquad
 \partial_y\Phi_k=-\sqrt{\frac{k+1}{2}}\Phi_{k+1}
 +\sqrt{\frac{k}{2}}\Phi_{k-1}
 \label{eq:ladder}
\end{equation}
give, on \(\operatorname{span}\{\Phi_0,\ldots,\Phi_{K-1}\}\),
\begin{equation*}
 \Ymat_K=\sum_{k=0}^{K-2}\sqrt{\frac{k+1}{2}}
 (\ket{k+1}\bra{k}+\ket{k}\bra{k+1}),\qquad
 \Dmat_{y,K}=\sum_{k=0}^{K-2}\sqrt{\frac{k+1}{2}}
 (\ket{k}\bra{k+1}-\ket{k+1}\bra{k}).
\end{equation*}
Let \(I_K=\{0,\ldots,K-1\}^d\) and \(\iota_K\ket k=\Phi_k\).
Define \(P_K=\iota_K\iota_K^\dagger\),
\(\iota_{\mathbf K}=\iota_K^{\otimes d}\),
\(\Pi_K=P_K^{\otimes d}\), and \(Q_K=I-\Pi_K\). The retained space
\(\mathcal H_K^{(d)}=\operatorname{ran}\Pi_K
=\operatorname{span}\{\Phi_{\mathbf k}:\mathbf k\in I_K\}\)
has dimension \(K^d\).
The tensor basis and cutoff do not require separability of the solution
or potential: mixed derivatives \(\partial^\nu V_{\rm P}\), monomials
\(y^\nu\), and Fourier couplings \(\xi\cdot y\) retain all coordinate
interactions.
The exact compressed moments are
\begin{equation*}
 \Ymat_K^{[r]}:=\iota_K^\dagger y^r\iota_K,\qquad
 \Ymat_{\mathbf K}^{[\nu]}
 :=\bigotimes_{j=1}^d\Ymat_K^{[\nu_j]}.
\end{equation*}
In general, \(\Ymat_K^{[r]}\ne(\Ymat_K)^r\): the latter inserts a
projection onto the retained modes between consecutive multiplications.

\subsection{Gauss--Hermite representation}\label{sec:GH}
We first describe the construction in one dimension. Let
\(\widetilde K=K+L_{\rm buf}\), where \(L_{\rm buf}\) is the number of
additional Hermite modes beyond the retained \(K\) modes. Let
\(\{y_q,\omega_q^{\rm H}\}_{q=0}^{\widetilde K-1}\) be the
Gauss--Hermite nodes and positive weights for the weight \(e^{-y^2}\):
\begin{equation*}
 \int_{\R}e^{-y^2}f(y)\,\dd y
 \approx\sum_{q=0}^{\widetilde K-1}\omega_q^{\rm H}f(y_q).
\end{equation*}
The rule is exact for polynomials of degree at most
\(2\widetilde K-1\), and the nodes are the zeros of
\(H_{\widetilde K}\). Define
\begin{equation*}
 p_k(y):=\frac{H_k(y)}{\pi^{1/4}\sqrt{2^k k!}},\qquad
 \Phi_k(y)=e^{-y^2/2}p_k(y).
\end{equation*}
Thus \(\int_{\R}e^{-y^2}p_k(y)p_l(y)\,\dd y=\delta_{kl}\).
At a fixed spatial point \(x_i\), the exact retained potential matrix
and its quadrature approximation are
\begin{align*}
 [U_{\ve,K}(x_i)]_{kl}
 &=\int_{\R}e^{-y^2}U_\ve(x_i,y)p_k(y)p_l(y)\,\dd y
\approx\sum_{q=0}^{\widetilde K-1}\omega_q^{\rm H}
 U_\ve(x_i,y_q)p_k(y_q)p_l(y_q),\qquad 0\le k,l<K.
\end{align*}

The multiplication identity in Eq.~\eqref{eq:ladder} gives the
symmetric tridiagonal Hermite ladder matrix
\begin{equation*}
 \Ymat_{\widetilde K}
 =\sum_{k=0}^{\widetilde K-2}\sqrt{\frac{k+1}{2}}
 (\ket{k+1}\bra{k}+\ket{k}\bra{k+1}).
\end{equation*}
For \(\mathbf v(y_q):=(p_0(y_q),\ldots,
 p_{\widetilde K-1}(y_q))^\top\), the recurrence and
\(H_{\widetilde K}(y_q)=0\) give
\(\Ymat_{\widetilde K}\mathbf v(y_q)=y_q\mathbf v(y_q)\).
Let \(S\in\R^{\widetilde K\times\widetilde K}\) have entries
\(S_{kq}:=\sqrt{\omega_q^{\rm H}}p_k(y_q)\).
Since \(\deg(p_kp_l)\le2\widetilde K-2\), quadrature exactness gives
\begin{equation*}
 (SS^\top)_{kl}
 =\sum_q\omega_q^{\rm H}p_k(y_q)p_l(y_q)=\delta_{kl}.
\end{equation*}
Consequently,
\begin{equation*}
 \Ymat_{\widetilde K}
 =S\operatorname{diag}(y_0,\ldots,y_{\widetilde K-1})S^\top,
 \qquad
 U_\ve(x_i,\Ymat_{\widetilde K})
 :=S\operatorname{diag}_{q}[U_\ve(x_i,y_q)]S^\top.
\end{equation*}
In particular,
\([U_\ve(x_i,\Ymat_{\widetilde K})]_{kl}
 =\sum_q\omega_q^{\rm H}U_\ve(x_i,y_q)p_k(y_q)p_l(y_q)\).
For any \(N\ge K\), let
\(J_{K,N}:\mathbb C^K\to\mathbb C^N\) be the canonical embedding,
\(J_{K,N}\ket{k}=\ket{k}\). The retained Gauss--Hermite approximation
is therefore
\begin{equation*}
 U_{\ve,K}(x_i)
 \approx J_{K,\widetilde K}^{\dagger}
 U_\ve(x_i,\Ymat_{\widetilde K})J_{K,\widetilde K}.
\end{equation*}
This matrix function equals the quadrature approximation; it need not
equal the exact projected multiplication operator.

For a monomial \(y^r\), however, an \(N\)-point rule is exact for
all retained matrix entries whenever
\(r+2K-2\le2N-1\), equivalently
\(N\ge K+\lfloor r/2\rfloor\). Applying the same diagonalization
at size \(N\) therefore gives the exact identity
\begin{equation}
 \Ymat_K^{[r]}=J_{K,N}^\dagger\Ymat_N^rJ_{K,N},\qquad
 N\ge K+\lfloor r/2\rfloor.
 \label{eq:GH-exact-moments}
\end{equation}
The polynomial block in Eq.~\eqref{eq:poly-block} uses these exact
moments, with an auxiliary size \(N\) large enough for each required
power. This choice of \(N\) is independent of the auxiliary buffer for
\(V_{\rm F}\): when \(V_{\rm F}=0\), this buffer may vanish while the
polynomial moments are still evaluated exactly. For the smooth
non-polynomial contribution, tensor products
of the one-dimensional construction give
\(J_d=J_{K,\widetilde K}^{\otimes d}\) and the finite matrix function
in Eq.~\eqref{eq:Ffs-proof}; its difference from the exact projected
operator is bounded in Eqs.~\eqref{eq:pathidentity}--\eqref{eq:buffer}.

\subsection{Finite differences and potential blocks}

For a banded matrix
\(A=\sum_\eta\sum_km_\eta(k)\ket{k+\eta}\bra k\), set
\begin{equation}
 \mathfrak b(A):=\sum_\eta\max_k|m_\eta(k)|.
 \label{eq:bdef}
\end{equation}
An \(r\)-step Hermite expansion has \(r+1\) parity-compatible net shifts and at most \(2^r\)
paths, while each ladder factor is bounded by \(\sqrt{(K+r)/2}\).
Hence
\begin{equation}
 \mathfrak b(\Ymat_K^{[r]})
 \le(r+1)2^r\left(\frac{K+r}{2}\right)^{r/2}
 \le C_rK^{r/2},\qquad
 \mathfrak b(\Ymat_{\mathbf K}^{[\nu]})\le
C_{|\nu|}K^{|\nu|/2}.
 \label{eq:moment-b}
\end{equation}
Here we use the tensor-product extension of
\(\mathfrak b\), for which
\(\mathfrak b(A\otimes B)=\mathfrak b(A)\mathfrak b(B)\). Hence
\(C_{|\nu|}\) depends only on the total degree \(|\nu|\), not on \(d\).

Let \(I_M=\{0,\ldots,M-1\}^d\) index a periodic grid of spacing \(h\)
on the box \(\Omega_x\). For an \(r_x\)th-order centered
difference of half-width \(L_x\), let
\begin{align}
 \Dmat_{x,h}
 &=h^{-1}\sum_{\ell=1}^{L_x}c_\ell^{(r_x)}(S_-^\ell-S_+^\ell),
 \qquad \lambda_x^{(r_x)}=2\sum_{\ell=1}^{L_x}|c_\ell^{(r_x)}|,
\qquad
 \Dmat_{x_j,h}
 :=\mathbb I_M^{\otimes(j-1)}\otimes\Dmat_{x,h}
 \otimes\mathbb I_M^{\otimes(d-j)}.
 \label{eq:stencil}
\end{align}
Here \(S_\pm\ket i=\ket{i\pm1\pmod M}\),
\(c_\ell^{(r_x)}\in\R\), and \(\mathbb I_n\) denotes the
\(n\times n\) identity matrix. The centered-difference coefficients
satisfy \(2\sum_{\ell=1}^{L_x}\ell c_\ell^{(r_x)}=1\) and the
order-\(r_x\) consistency conditions, so that \(\Dmat_{x,h}\)
approximates \(+\partial_x\). For retained indices
\(0\le k_j<K\), set
\begin{equation*}
 R_{\mathbf k}(t,x)=\int_{\R^d}R(t,x,y)\Phi_{\mathbf k}(y)\,\dd y,
 \qquad
 (\mathcal S_hR)_{\mathbf i,\mathbf k}=h^{d/2}R_{\mathbf k}(x_{\mathbf i}).
\end{equation*}
Thus \(\mathcal S_hR\in\mathbb C^{M^dK^d}\), and
\(\mathbf R_0\) contains the sampled initial coefficients.
Section~\ref{sec:proof2} specifies the restricted and periodized initial
data and the associated restriction, sampling, and normalization errors.
The tensor lift is
\(\Dmat_{y_j,K}=\mathbb I_K^{\otimes(j-1)}\otimes\Dmat_{y,K}
\otimes\mathbb I_K^{\otimes(d-j)}\).
Then
\begin{equation}
 H_{\rm tr}=-\sum_{j=1}^d\Dmat_{x_j,h}\otimes\Dmat_{y_j,K},\qquad
 \mathfrak b(\Dmat_{y,K})\le\sqrt{2K}.
 \label{eq:transport}
\end{equation}

For the retained indices \(\mathbf k,\mathbf l\in I_K\), define
\begin{equation*}
 [U_{\ve,K}(x)]_{\mathbf k\mathbf l}
 :=\int_{\R^d}U_\ve(x,y)
 \Phi_{\mathbf k}(y)\Phi_{\mathbf l}(y)\,\dd y,
 \qquad
 U^{\rm ex}_{\ve,h,K}
 :=\sum_{\mathbf i\in I_M}\ket{\mathbf i}\bra{\mathbf i}
 \otimes U_{\ve,K}(x_{\mathbf i}).
\end{equation*}
Thus the exact semidiscrete evolution is
\(\ii\partial_t\mathbf R=(H_{\rm tr}+U^{\rm ex}_{\ve,h,K})\mathbf R\).

For the split \(V=V_{\rm P}+V_{\rm F}\), with
\(V_{\rm P}\) of fixed total degree \(D_{\rm P}\) and
\(V_{\rm F}\) a smooth non-polynomial potential, write
\(U_\ve=U_\ve^{\rm P}+U_\ve^{\rm F}\). Set
\(\qP=\lfloor(D_{\rm P}-1)/2\rfloor\). The exact polynomial expansion is
\begin{equation}
 U_\ve^{\rm P}(x,y)
 =\sum_{m=0}^{\qP}\frac{\ve^{2m}}{2^{2m}}
 \sum_{|\nu|=2m+1}
 \frac{(\partial^\nu V_{\rm P})(x)y^\nu}{\nu!}.
 \label{eq:poly-U-supp}
\end{equation}
Consequently, its exact retained block is
\begin{equation}
 U^{\rm P}_{\ve,h,K}
 =\sum_{m=0}^{\qP}\frac{\ve^{2m}}{2^{2m}}
 \sum_{|\nu|=2m+1}\frac{\mathsf M_{\nu,V_{\rm P}}
 \otimes\Ymat_{\mathbf K}^{[\nu]}}{\nu!},
 \label{eq:poly-block}
\end{equation}
where \(\mathsf M_{\nu,V_{\rm P}}\ket{\mathbf i}
=(\partial^\nu V_{\rm P})(x_{\mathbf i})\ket{\mathbf i}\) and
\(\qP\) is as above.
For \(V_{\rm P}=0\), take \(D_{\rm P}=0\); all polynomial sums are
empty when \(D_{\rm P}=0\). We use
\(\Gamma_{m,d}^{\rm P}(\Omega_x)=\sum_{|\nu|=2m+1}
\norm{\partial^\nu V_{\rm P}}_{L^\infty(\Omega_x)}/\nu!\).

For \(V_{\rm F}(x)=(2\pi)^{-d}\int_{\R^d}\widehat V_{\rm F}(\xi)
e^{\ii\xi\cdot x}\,\dd\xi\), where \(\widehat V_{\rm F}\) denotes
the Fourier transform, choose paired nodes \(\xi_{-q}=-\xi_q\) and
positive weights \(w_{-q}=w_q\), and set
\(a_q=w_q\widehat V_{\rm F}(\xi_q)/(2\pi)^d\).
For a periodic potential, \(a_q\) is the corresponding Fourier
coefficient. Omitting the dynamically irrelevant zero mode gives
\begin{equation*}
 V_{\mathcal Q}^{\rm F}(x)=\sum_{q\in\mathcal Q}a_qe^{\ii\xi_q\cdot x},
 \qquad \xi_q\ne0,\qquad a_{-q}=\overline{a_q},
\end{equation*}
where conjugate pairing follows from the reality of \(V_{\rm F}\).
The Fourier weights \(w_q\) are distinct from the Gauss--Hermite
weights \(\omega_q^{\rm H}\) of Sec.~\ref{sec:GH}.
With \(\operatorname{sinc}z=\sin(z)/z\), continuously extended at
\(z=0\), define
\begin{align*}
 (\mathcal Y_jf)(y)&=y_jf(y),\qquad
 \mathcal A_q=\sum_{j=1}^d\xi_{q,j}\mathcal Y_j,\qquad
 f_\ve(z)=z\operatorname{sinc}(\ve z/2),\\
 \mathsf F_{q,K}^{\rm full}&=\iota_{\mathbf K}^\dagger f_\ve(\mathcal A_q)
 \iota_{\mathbf K},\qquad
 \mathsf E_q=\sum_{\mathbf i}e^{\ii\xi_q\cdot x_{\mathbf i}}
 \ket{\mathbf i}\!\bra{\mathbf i}.
\end{align*}
The exact Fourier target for the selected quadrature is
\begin{equation*}
 U^{\rm F,tar}_{\ve,h,K}=\ii\sum_{q\in\mathcal Q}
 a_q\mathsf E_q\otimes \mathsf F_{q,K}^{\rm full}.
\end{equation*}
Let \(\mathcal M_g\) denote multiplication by \(g\). The exact
retained-sector compression of the untruncated component is
\begin{equation*}
 U^{\rm F,ex}_{\ve,h,K}
 =\sum_{\mathbf i\in I_M}\ket{\mathbf i}\!\bra{\mathbf i}\otimes
 \iota_{\mathbf K}^\dagger
 \mathcal M_{U_\ve^{\rm F}(x_{\mathbf i},\cdot)}
 \iota_{\mathbf K}.
\end{equation*}

To approximate \(\mathsf F_{q,K}^{\rm full}\), use the enlarged size
\(\widetilde K=K+L_{\rm buf}\) and the tensor embedding
\(J_d=J_{K,\widetilde K}^{\otimes d}:\mathbb C^{K^d}
\to\mathbb C^{\widetilde K^d}\) from Sec.~\ref{sec:GH}. Define
\begin{align}
 \Ymat_{j,\widetilde K}
 =\mathbb I_{\widetilde K}^{\otimes(j-1)}\otimes\Ymat_{\widetilde K}
 \otimes \mathbb I_{\widetilde K}^{\otimes(d-j)},
 \qquad
 \mathsf A_{q,\widetilde K}=\sum_j\xi_{q,j}\Ymat_{j,\widetilde K},
 \nonumber\\
 \beta_q=\sqrt{2\widetilde K}\norm{\xi_q}_1,
 \qquad \beta_{\max}=\max_{q\in\mathcal Q}\beta_q,
 \quad\max\varnothing:=0.
 \label{eq:Atilde}
\end{align}

Set
\begin{equation}
 \mathsf F_{q,K}^{\rm fs}
 :=J_d^\dagger f_\ve(\mathsf A_{q,\widetilde K})J_d,
 \qquad
 U^{\rm F,fs}_{\ve,h,K}
 :=\ii\sum_{q\in\mathcal Q}a_q\mathsf E_q\otimes\mathsf F_{q,K}^{\rm fs}.
 \label{eq:Ffs-proof}
\end{equation}
These are the finite-section matrices used in the Letter.
The matrix \(\mathsf F_{q,K}^{\rm full}\) is the exact projected Fourier
matrix used to quantify the finite-section error.

For a real odd polynomial used in quantum singular value transformation
(QSVT), denoted by \(p_{q,\ve}\), set
\begin{equation*}
 \widetilde{\mathsf F}_{q,K}:=J_d^\dagger\beta_q
 p_{q,\ve}(\mathsf A_{q,\widetilde K}/\beta_q)J_d,\qquad
 U^{\rm F,QSVT}_{\ve,h,K}:=\ii\sum_qa_q\mathsf E_q\otimes
 \widetilde{\mathsf F}_{q,K}.
\end{equation*}
At finite precision, paired data are generated from one representative
by exact conjugation or sign reversal:
\(\widetilde a_{-q}=\overline{\widetilde a_q}\),
\(\widetilde\xi_{-q}=-\widetilde\xi_q\), and
\(\widetilde{\mathsf E}_{-q}=\widetilde{\mathsf E}_q^\dagger\).
The paired modes share a real odd QSVT polynomial and adjoint-paired
circuits, so the implemented generator for \(V_{\rm F}\) remains Hermitian.
Alternatively, taking its Hermitian part cannot increase the
operator-norm error from the exact Hermitian target.
The aggregate physical operator defect from finite-precision
coefficients, oracle values, and phases is denoted by \(\eta_{\rm arith}\).

For the proof, distinguish the exact Hamiltonian, the target using
selected Fourier data with exact Hermite projection, and the implemented
Hamiltonian:
\begin{align}
 H^{\rm ex}_{\ve,h,K}
 &=H_{\rm tr}+U^{\rm P}_{\ve,h,K}+U^{\rm F,ex}_{\ve,h,K},
 \qquad
 H^{\rm tar}_{\ve,h,K}
 =H_{\rm tr}+U^{\rm P}_{\ve,h,K}+U^{\rm F,tar}_{\ve,h,K},
\qquad
 H_{\rm WH}
 =H_{\rm tr}+U^{\rm P}_{\ve,h,K}+U^{\rm F,QSVT}_{\ve,h,K},
 \nonumber\\
 \eta_{\rm F}&=\norm{H^{\rm ex}_{\ve,h,K}-H^{\rm tar}_{\ve,h,K}},
 \qquad
 \eta_{\rm syn}=\norm{H^{\rm tar}_{\ve,h,K}-H_{\rm WH}},
 \qquad
 \eta_{\rm pot}=\norm{H^{\rm ex}_{\ve,h,K}-H_{\rm WH}}
 \le\eta_{\rm F}+\eta_{\rm syn}.
 \label{eq:Hlayers}
\end{align}
Indeed,
\(\Dmat_{x_j,h}^\dagger=-\Dmat_{x_j,h}\) and
\(\Dmat_{y_j,K}^\dagger=-\Dmat_{y_j,K}\), the exact Hermite moments are real
symmetric, and the Fourier modes are conjugate-paired. Hence all three
terms in \(H_{\rm WH}\) are Hermitian.

Here \(U^{\rm ex}_{\ve,h,K}=U^{\rm P}_{\ve,h,K}+U^{\rm F,ex}_{\ve,h,K}\).
The intermediate target separates Fourier quadrature error from
finite-section and QSVT implementation errors.

For a real Weyl symbol \(a\), set
\begin{align}
 (\mathbf b_a)_{\mathbf i,\mathbf k}
 &=h^{d/2}\int_{\R^d}\check a(x_{\mathbf i},y)
 \Phi_{\mathbf k}(y)\,\dd y,
 \qquad
 A^\ve(T)=\langle\check a,R^\ve(T)\rangle_{L^2_{x,y}},
 \qquad
 A_{\rm WH}(T)=\langle\mathbf b_a,e^{-\ii TH_{\rm WH}}\mathbf R_0\rangle.
 \label{eq:observables}
\end{align}
We write \(s_R=\norm{\mathbf R_0}_2\) and \(s_a=\norm{\mathbf b_a}_2\).
For the uniform observable guarantees, we assume
\(s_R,s_a,\norm{\check a}_{L^2_{x,y}}=O(1)\), uniformly in
\(\ve\) and \(d\), along the accuracy-selected family.
These input--readout assumptions are distinct from the fixed-discretization
block-encoding bounds. State preparation and overlap estimation are
addressed in Secs.~\ref{sec:proof3} and \ref{sec:io}.

\section{Simulation at fixed discretization}
\label{sec:proof1}

\begin{lemma}[Simulation at fixed discretization]
\label{lem:fixed-resources}
Fix \(T>0\), the stencil order \(r_x\), the polynomial degree
\(D_{\rm P}\), and the discretization parameters
\(h,K,\mathcal Q,L_{\rm buf}\) of Sec.~\ref{sec:finite}.
Assume reversible evaluation of polynomial derivatives, Fourier
coefficients, frequencies, and grid phases, coherent preparation of the
Fourier weights, and access to the QSVT phases selected by the Fourier
index, with primitive costs uniform in \(\ve\).
Let \(\widetilde K=K+L_{\rm buf}\) and
\(\mathcal V_{1,\mathcal Q}=\sum_{q\in\mathcal Q}|a_q|\norm{\xi_q}_1\).
Then the block encoding of \(H_{\rm WH}\) has normalization
\begin{align}
 \alpha_{\rm WH}
 =O_{r_x,D_{\rm P}}\!\Big(&dh^{-1}K^{1/2}
 +\sum_{m=0}^{\qP}\ve^{2m}K^{m+1/2}
 \Gamma_{m,d}^{\rm P}(\Omega_x)+\sqrt{\widetilde K}\,\mathcal V_{1,\mathcal Q}\Big).
 \label{eq:lemma-fixed-alpha}
\end{align}
For \(0<\delta_{\rm sim}<1\), simulation with operator-norm error
at most \(\delta_{\rm sim}\) uses
\begin{equation}
 N_{\rm sim}=\wtO\!\left(T\alpha_{\rm WH}
 +\log\delta_{\rm sim}^{-1}\right)
 \label{eq:lemma-fixed-simulation}
\end{equation}
queries to this block encoding.
\end{lemma}

The notation \(\wtO\) suppresses logarithmic factors. If the Fourier set
is nonempty, each full-Hamiltonian block query uses \(O(n_{\rm F})\)
queries to the Hermite coordinate block encodings, with \(n_{\rm F}\)
given in Eq.~\eqref{eq:nF}. For an empty Fourier set, these inner queries
are omitted.

\begin{proof}
Using the stated access assumptions, direct reversible data access and
standard arithmetic prepare
\begin{equation*}
 \ket q\ket0\longmapsto\ket q\sum_{j=1}^d
 \sqrt{\frac{|\xi_{q,j}|}{\norm{\xi_q}_1}}\,
 \ket j\ket{\operatorname{sgn}(\xi_{q,j})}
\end{equation*}
with \(O(d)\) controlled rotations at the selected precision; a more
succinct data oracle may reduce this cost. All primitive costs are uniform
in \(\ve\) and enter the gate bound in Section~\ref{sec:proof3}.

For the shifted-diagonal representation in Eq.~\eqref{eq:bdef}, the
standard PREPARE/SELECT construction block encodes \(A\) with
normalization \(\mathfrak b(A)\), using one controlled modular shift and
one controlled amplitude rotation per selected diagonal
\cite{sm:Gilyen2019}; the tensor-product construction is identical.

From Eqs.~\eqref{eq:stencil}, \eqref{eq:transport}, and \eqref{eq:moment-b},
\begin{align*}
 \alpha_{\rm tr}
 &=d\lambda_x^{(r_x)}h^{-1}\mathfrak b(\Dmat_{y,K})
 =O_{r_x}(dh^{-1}K^{1/2}),\\
 \alpha_{\rm pot}^{\rm P}
 &\le\sum_{m=0}^{\qP}\frac{\ve^{2m}}{2^{2m}}
 \sum_{|\nu|=2m+1}
 \frac{\norm{\partial^\nu V_{\rm P}}_{L^\infty(\Omega_x)}}{\nu!}
 \mathfrak b(\Ymat_{\mathbf K}^{[\nu]})
=O_{D_{\rm P}}\!\left(\sum_{m=0}^{\qP}
 \ve^{2m}K^{m+1/2}\Gamma_{m,d}^{\rm P}(\Omega_x)\right).
\end{align*}

By the multinomial theorem,
\begin{align}
 \Gamma_{m,d}^{\rm P}(\Omega_x)
 &\le \frac{d^{2m+1}}{(2m+1)!}
 \max_{|\nu|=2m+1}
 \norm{\partial^\nu V_{\rm P}}_{L^\infty(\Omega_x)}.
 \label{eq:Gamma-d-bound}
\end{align}
For fixed \(D_{\rm P}\), the number of polynomial multi-indices is
\(N_{{\rm mi},d}=\sum_{m=0}^{\qP}\binom{d+2m}{2m+1}\), hence polynomial
in \(d\).

The linear combination of unitaries (LCU) for the shifted diagonals
of \(\mathsf A_{q,\widetilde K}\) has
normalization \(\beta_q\). Set \(\theta_q=\ve\beta_q/2\) and
\begin{equation*}
 g_{q,\ve}(s)=s\operatorname{sinc}(\theta_qs)
 =\begin{cases}\sin(\theta_qs)/\theta_q,&\theta_q\ne0,\\
 s,&\theta_q=0,
 \end{cases}
 \qquad |g_{q,\ve}(s)|\le1\quad(s\in[-1,1]).
\end{equation*}

For \(0<\eta\le1/4\), the standard Jacobi--Anger truncation,
followed by a \(1+O(\eta)\) rescaling and a constant-factor tightening of
the truncation tolerance, together with the usual Bessel
tail bound, yields a real odd polynomial satisfying
\begin{equation}
 \norm{p_{q,\ve}}_{L^\infty[-1,1]}\le1,\qquad
 \norm{p_{q,\ve}-g_{q,\ve}}_{L^\infty[-1,1]}\le\eta,
 \qquad
 \deg p_{q,\ve}=O(\ve\beta_q+\log\eta^{-1}).
 \label{eq:qsvtdegree}
\end{equation}
The estimate remains uniform as \(\theta_q\to0\), as follows directly
from the power series of \(J_{2m+1}(\theta_q)/\theta_q\).

QSVT therefore block encodes
\(f_\ve(\mathsf A_{q,\widetilde K})/\beta_q\) with normalized error
\(\eta\), equivalently with physical operator error
\(\beta_q\eta\) \cite{sm:LowChuang2017,sm:Gilyen2019}.

Finite Hermite combinations are analytic vectors for
\(\mathcal A_q\), so the following entire series may be compressed
termwise on \(\operatorname{ran}\Pi_K\). The finite-section error follows
from
\begin{equation}
 f_\ve(A)=\sum_{m\ge0}
 \frac{(-1)^m\ve^{2m}A^{2m+1}}{2^{2m}(2m+1)!}.
 \label{eq:sincseries}
\end{equation}

Each factor of \(\mathcal A_q\) changes exactly one Hermite index by
\(\pm1\). A path starting and ending in \(I_K\) that
visits an omitted level \(k_j\ge\widetilde K=K+L_{\rm buf}\) must contain
at least \(L_{\rm buf}+1\) raising and \(L_{\rm buf}+1\) lowering
steps in the same coordinate. Hence no odd path of degree
\(r\le2L_{\rm buf}+1\) detects the boundary, and the first possibly
different odd power is \(R:=2L_{\rm buf}+3\):

\begin{equation}
 \iota_{\mathbf K}^\dagger\mathcal A_q^r\iota_{\mathbf K}
 =J_d^\dagger\mathsf A_{q,\widetilde K}^rJ_d,\qquad
 r=1,3,\ldots,2L_{\rm buf}+1.
 \label{eq:pathidentity}
\end{equation}
The ladder bounds
\begin{equation*}
 \norm{\mathcal A_q^r\Pi_K}
 \le[\norm{\xi_q}_1\sqrt{2(K+r)}]^r,\qquad
 \norm{\mathsf A_{q,\widetilde K}}\le\beta_q
\end{equation*}
and Eq.~\eqref{eq:sincseries} yield
\begin{align}
 \delta_{q,K}^{\rm fs}
 &:=\norm{\mathsf F_{q,K}^{\rm full}-\mathsf F_{q,K}^{\rm fs}}%
 \le\sum_{\substack{r=2m+1>2L_{\rm buf}+1}}
 \frac{\ve^{r-1}}{2^{r-1}r!}
 \left([\norm{\xi_q}_1\sqrt{2(K+r)}]^r+\beta_q^r\right).
 \label{eq:fserror}
\end{align}

Write \(B=\norm{\xi_q}_1\), \(R=2L_{\rm buf}+3\), and denote the two
summands in Eq.~\eqref{eq:fserror} by \(u_r\) and \(v_r\).
For odd \(r\ge R\), direct division gives
\begin{equation*}
 \max\!\left\{\frac{u_{r+2}}{u_r},\frac{v_{r+2}}{v_r}\right\}
 \le C\ve^2B^2\frac{K+r}{r^2}.
\end{equation*}
Thus \(R\ge C_0(\ve^2B^2+\ve B\sqrt K)\), with
\(C_0\) sufficiently large, makes both ratios at most a universal
\(\rho<1\). Stirling's inequality bounds the first tail terms by
\(C(1+B\sqrt K)e^{-cR}\); summing the two geometric odd tails gives
\begin{equation*}
 \delta_{q,K}^{\rm fs}
 \le\frac{u_R+v_R}{1-\rho}
 \le C(1+B\sqrt K)e^{-cL_{\rm buf}}
 \quad\text{if}\quad
 L_{\rm buf}\ge C(\ve^2B^2+\ve B\sqrt K).
\end{equation*}
Here \(C,c,C_0,\rho\) are independent of
\(d,q,K,\ve\); the dimension enters only through
\(B=\norm{\xi_q}_1\). Consequently, for
\(0<\eta_{{\rm fs},q}\le1\), a convenient sufficient condition is
\begin{equation*}
 L_{\rm buf}\ge C\!\left[
 \ve\norm{\xi_q}_1\sqrt{K+L_{\rm buf}}
 +\log\frac{1+\norm{\xi_q}_1\sqrt K}{\eta_{{\rm fs},q}}
 \right].
\end{equation*}
Using \(2ab\le a^2+b^2\) to solve for \(L_{\rm buf}\) yields
\begin{equation}
 L_{\rm buf}\ge C\!\left[
 \ve^2\norm{\xi_q}_1^2+\ve\norm{\xi_q}_1\sqrt K
 +\log\frac{1+\norm{\xi_q}_1\sqrt K}{\eta_{{\rm fs},q}}
 \right]
\Longrightarrow\delta_{q,K}^{\rm fs}\le\eta_{{\rm fs},q}.
 \label{eq:buffer}
\end{equation}

If the selected set \(\mathcal Q\) is empty after zero-weight
modes are discarded, we omit the Fourier block and set
\(A_0=\alpha_{\rm pot}^{\rm F}=\beta_{\max}=n_{\rm F}=0\). Otherwise
\(A_0>0\), and for target budgets
\(\eta_{\rm fs},\eta_{\rm QSVT}>0\), let
\begin{equation*}
 A_0=\sum_{q\in\mathcal Q}|a_q|,\qquad
 \overline\eta_{\rm fs}=
 \min\!\left\{\frac14,\frac{\eta_{\rm fs}}{A_0}\right\},
 \qquad \eta_{{\rm fs},q}=\overline\eta_{\rm fs}.
\end{equation*}
A common buffer satisfying Eq.~\eqref{eq:buffer} for every mode then gives
\begin{equation*}
 \sum_q|a_q|\delta_{q,K}^{\rm fs}
 \le A_0\overline\eta_{\rm fs}\le\eta_{\rm fs}.
\end{equation*}
\begin{equation}
 \alpha_{\rm pot}^{\rm F}
 =\sum_{q\in\mathcal Q}|a_q|\beta_q
 =\sqrt{2\widetilde K}\,\mathcal V_{1,\mathcal Q},\qquad
 \mathcal V_{1,\mathcal Q}:=\sum_q|a_q|\norm{\xi_q}_1.
 \label{eq:aF}
\end{equation}
The outer Fourier LCU then prepares
\begin{equation*}
 \operatorname{PREPARE}_{\rm F}\ket0
 =\frac1{\sqrt{\alpha_{\rm pot}^{\rm F}}}
 \sum_{q\in\mathcal Q}\sqrt{|a_q|\beta_q}\ket q .
\end{equation*}
SELECT applies the coefficient phase, \(\mathsf E_q\), the global factor
\(\ii\), and the mode-multiplexed QSVT sequence. For a common QSVT
tolerance \(\eta\),
\begin{align*}
 \norm{\mathsf F_{q,K}^{\rm full}-\widetilde{\mathsf F}_{q,K}}
 &\le\delta_{q,K}^{\rm fs}+\beta_q\eta
 +\delta_q^{\rm arith},
 \quad
 \sum_q|a_q|\delta_{q,K}^{\rm fs}\le\eta_{\rm fs},
 \quad
 \sum_q|a_q|\delta_q^{\rm arith}\le\eta_{\rm arith},
 \quad
 \eta=\min\!\left\{\frac14,
 \frac{\eta_{\rm QSVT}}{\alpha_{\rm pot}^{\rm F}}\right\}.
\end{align*}
Here \(\delta_q^{\rm arith}\) includes finite-precision
errors in the modewise QSVT phases, Fourier coefficients and frequencies,
coordinate block encodings, and grid phases. The primitive precisions are
tightened according to the relevant QSVT degree so that robust QSVT error
propagation places their weighted sum within the aggregate physical
operator budget \(\eta_{\rm arith}\).
Consequently
\(\eta_{\rm syn}\le
\eta_{\rm fs}+\eta_{\rm QSVT}+\eta_{\rm arith}\), and
Eq.~\eqref{eq:Hlayers} gives
\begin{equation*}
 \eta_{\rm pot}\le
 \eta_{\rm F}+\eta_{\rm fs}+\eta_{\rm QSVT}+\eta_{\rm arith}.
\end{equation*}
Here \(\eta_{\rm arith}\) belongs to the selected potential-generator
approximation, whereas the later \(\eta_{\rm BE}\) denotes only a
further full-Hamiltonian block-encoding defect.
We use \(\delta_{\rm pot}\) with the same meaning as in the Letter:
it bounds the error of the implemented Fourier block relative to the
finite-section target,
\begin{equation*}
 \norm{U_{\ve,h,K}^{{\rm F},{\rm QSVT}}
       -U_{\ve,h,K}^{{\rm F},{\rm fs}}}
 \le\eta_{\rm QSVT}+\eta_{\rm arith}\le\delta_{\rm pot}.
\end{equation*}
Choose \(\eta_{\rm QSVT}=\delta_{\rm pot}/2\) and
\(\eta_{\rm arith}\le\delta_{\rm pot}/2\).
The Fourier-data and finite-section errors are budgeted separately, so
the total potential-generator defect satisfies
\begin{equation*}
 \eta_{\rm pot}\le\eta_{\rm F}+\eta_{\rm fs}+\delta_{\rm pot}.
\end{equation*}
Taking \(\eta_{\rm F},\eta_{\rm fs}=O(\delta_{\rm pot})\) therefore
gives \(\eta_{\rm pot}=O(\delta_{\rm pot})\).
Thus a single outer LCU, in which the transport contribution is included
once, has normalization

\begin{align}
 \alpha_{\rm WH}(\ve;h,K,\mathcal Q,L_{\rm buf})
 =O_{r_x,D_{\rm P}}\!\Big(&dh^{-1}K^{1/2}
 +\sum_{m=0}^{\qP}\ve^{2m}K^{m+1/2}\Gamma_{m,d}^{\rm P}(\Omega_x)
 +\sqrt{\widetilde K}\,\mathcal V_{1,\mathcal Q}\Big).
 \label{eq:alpha-fixed}
\end{align}

For fixed \(h,K,\mathcal Q,L_{\rm buf}\) and the selected QSVT
polynomials, let \(\delta_{\rm sim}\) be the simulation tolerance
and \(\eta_{\rm QSP}\) the quantum signal processing (QSP) error. If a
finite-precision block encoding represents a Hermitian
\(\widetilde H_{\rm WH}\) with
physical operator defect
\(\norm{\widetilde H_{\rm WH}-H_{\rm WH}}\le\eta_{\rm BE}\), choose
\begin{equation*}
 \eta_{\rm QSP}=\frac12\delta_{\rm sim},
 \qquad T\eta_{\rm BE}\le\frac12\delta_{\rm sim}.
\end{equation*}
Duhamel's identity then bounds the total implemented-evolution error by
\(\delta_{\rm sim}\), and qubitization uses
\begin{equation}
 N_{\rm sim}=\wtO\!\left(
 T\alpha_{\rm WH}(\ve;h,K,\mathcal Q,L_{\rm buf})
 +\log\delta_{\rm sim}^{-1}\right)
 \label{eq:sim}
\end{equation}
full-Hamiltonian block queries \cite{sm:LowChuang2017,sm:LowChuang2019}.
When \(\alpha_{\rm pot}^{\rm F}>0\), one Fourier query invokes
\begin{equation}
 n_{\rm F}=O\!\left(1+\ve\beta_{\max}
 +\log\!\left[1+\frac{\alpha_{\rm pot}^{\rm F}}
 {\delta_{\rm pot}}\right]\right)
 \label{eq:nF}
\end{equation}
coordinate-block queries. Equations~\eqref{eq:alpha-fixed}--\eqref{eq:nF}
prove Lemma~\ref{lem:fixed-resources}, including the empty-Fourier case handled above.

\end{proof}

\section{Uniform parameter selection}
\label{sec:proof2}

\begin{lemma}[Uniform parameter selection]
\label{lem:uniform-selection}
Let the initial data and the potential satisfy the conditions for the
weighted-regularity propagation estimate specified below. Assume the
domain-restriction, periodization, sampling, spatial-stencil, and
input--readout conditions detailed below and in Sec.~\ref{sec:io}. Assume that the absolute zeroth
and first Fourier moments are finite. For the polynomial-in-\(d\)
conclusions, impose also the dimension-growth and structured-access
conditions specified in this section and Sec.~\ref{sec:proof3}; in
particular, the regularity constants and
\(\Gamma_{m,d}^{\rm P}(\Omega_{x,\delta})\) grow at most polynomially in
\(d\), and the required Fourier-tail accuracy is attained with
polynomial bandwidth at fixed \(T,\delta,D_{\rm P}\).

Fix \(T>0\) and \(p\in\mathbb N\), \(p\ge1\), with \(2p\ge r_x\),
as permitted by the regularity of the initial data and the potential. Let
\(\mathfrak T_{{\rm F},d}\) be the first absolute Fourier-moment tail
defined in Eq.~\eqref{eq:Fourier-tail-supp}. For each \(d\) and
\(0<\delta<1\), choose
\begin{equation}
 K_{d,\delta}
 =O\!\left(\operatorname{poly}(d)\,\delta^{-2/(2p-1)}\right).
 \label{eq:lemma-K-rate}
\end{equation}
Choose a paired Fourier quadrature with positive weights and
\(\norm{\xi_q}_\infty\le\Xi_{d,\delta}\), so that
\(\sqrt{K_{d,\delta}}\mathfrak T_{{\rm F},d}(\Xi_{d,\delta})
=O(\delta/[T\operatorname{poly}(d)])\), with quadrature error of
the same order in operator norm. Assume its discrete moments satisfy
\(\sum_q|a_q|+\mathcal V_{1,d,\delta}\le\operatorname{poly}(d)\),
uniformly in \(\delta,\ve\), where
\(\mathcal V_{1,d,\delta}=\sum_{q\in\mathcal Q_{d,\delta}}
|a_q|\norm{\xi_q}_1\).

The box \(\Omega_{x,\delta}\), mesh \(h_{d,\delta}\), and additional
Hermite modes can then be chosen independently of \(\ve\), with
\begin{align}
 L_{{\rm buf},d,\delta}=O\!\Big(&d^2\Xi_{d,\delta}^2
 +d\Xi_{d,\delta}\sqrt{K_{d,\delta}}+\log\!\left[1+\frac{T\operatorname{poly}(d)}{\delta}
 \bigl(1+d\Xi_{d,\delta}\sqrt{K_{d,\delta}}\bigr)\right]\Big).
 \label{eq:lemma-Lbuf-selection}
\end{align}
Set \(\widetilde K_{d,\delta}=K_{d,\delta}+L_{{\rm buf},d,\delta}\).
Selecting \(h_{d,\delta}^{-1}\ge\Xi_{d,\delta}\) and the remaining
implementation tolerances gives
\begin{equation}
 \sup_{0<\ve\le1}|A_{\rm WH}(T)-A^\ve(T)|=O(\delta).
 \label{eq:lemma-uniform-consistency}
\end{equation}
For fixed \(T,\delta\), the selected inverse mesh, bandwidth, and
Hermite sizes are polynomial in \(d\) under the stated dimension
assumptions. If \(V_{\rm F}=0\), take
\(\mathcal Q_{d,\delta}=\varnothing\) and
\(\Xi_{d,\delta}=\mathcal V_{1,d,\delta}=L_{{\rm buf},d,\delta}=0\).
\end{lemma}

\begin{proof}[Proof and detailed assumptions]
\noindent\emph{Error decomposition.---}
For the same observable, let \(A_{\rm box}^\ve(T)\),
\(A_K^\ve(T)\), and \(A_{\ve,h,K}^{\rm ex}(T)\) denote its
time-\(T\) values for the restricted--periodized,
Hermite--Galerkin, and exact semidiscrete evolutions,
respectively, and let \(\widetilde A(T)\) be the final quantum
estimate. The complete proof follows the chain
\begin{equation}
A^\ve(T)\longrightarrow A_{\rm box}^\ve(T)
 \longrightarrow A_K^\ve(T)
 \longrightarrow A_{\ve,h,K}^{\rm ex}(T)
 \longrightarrow A_{\rm WH}(T)\longrightarrow\widetilde A(T) .
 \label{eq:proof-map}
\end{equation}
All continuum pairings in the first two arrows use the
restricted observable. Moreover,
\(R_K^\ve(0)=\Pi_KR_{\rm box}^\ve(0)\) and
\(\mathbf R_0=\mathcal S_hR_K^\ve(0)\). The within-box difference
between the continuum pairing defining \(A_K^\ve(T)\) and the sampled
pairing defining \(A_{\ve,h,K}^{\rm ex}(T)\), including sampling,
quadrature, and normalization, is assigned to
\(\delta_{{\rm io},d,\delta}\) below.
The arrows represent, in order, restriction and periodization, Hermite
truncation, spatial discretization, potential synthesis, and Hamiltonian
simulation with readout.  This section controls the deterministic arrows
through \(A_{\rm WH}\); Section~\ref{sec:proof3} controls the last one.
For fixed \(d\), all analytic and data assumptions below are uniform in
\(0<\ve\le1\), but their constants need not be uniform in \(d\).
The polynomial-in-\(d\) statement additionally requires the listed
parameters and reversible access and preparation costs to be polynomial
in \(d\).

For \(\mathcal X=\mathbb R^d\) or an accuracy-selected periodic box
\(\Omega_{x,\delta}\), define
\begin{equation}
 \mathcal N_{m,\mathcal X}^{(d)}[F]
 :=\!\!\sum_{\substack{\alpha,\beta,\gamma,\zeta\in\mathbb N_0^d\\
 |\alpha|+|\beta|+|\gamma|+|\zeta|\le m}}
 \norm{x^\alpha\partial_x^\beta y^\gamma\partial_y^\zeta F}
 _{L^2(\mathcal X\times\mathbb R_y^d)},
 \qquad \mathcal N_m^{(d)}:=\mathcal N_{m,\mathbb R^d}^{(d)}.
 \label{eq:Nm}
\end{equation}
Let \(\mathcal L_\ve=\ii\nabla_x\!\cdot\!\nabla_y-\ii U_\ve\).
Fix \(p\in\mathbb N\) with \(p\ge1\) and \(2p\ge r_x\), and set
\(m_*=2p+\max\{2,D_{\rm P}-1\}\). For \(d=1\), the required
weighted regularity follows from Filbet and Golse
\cite[Proposition 4.1]{sm:FilbetGolse2026}. Let \(V\) satisfy the potential
conditions of that proposition at order \(m_*\), and assume that the
initial data satisfy
\begin{equation}
 C_{0,m_*}:=\sup_{0<\ve\le1}
 \mathcal N_{m_*}^{(1)}[R_0^\ve]<\infty.
 \label{eq:initial-weighted-regularity}
\end{equation}
The propagation estimate then gives
\begin{equation}
 \mathfrak N_{m_*,T}^{(1)}
 :=\sup_{\substack{0<\ve\le1\\0\le t\le T}}
 \mathcal N_{m_*}^{(1)}[R^\ve(t)]
 \le C_{0,m_*}e^{C_{m_*}[V]T},
 \label{eq:mstar}
\end{equation}
where \(C_{m_*}[V]\) is independent of \(\ve\).
The same bound is assumed for the periodized solutions on the selected
boxes; its supremum over \(0<\delta<1\), \(0<\ve\le1\), and
\(0\le t\le T\) is denoted by
\(\mathfrak N_{m_*,T,{\rm box}}^{(d)}\).
For \(d>1\), both the bound and its dependence on \(d\) remain explicit
hypotheses. Here \(R_{\rm box}^\ve\) solves the selected periodic-box
equation with restricted--periodized initial data. On one fundamental box
we first form \(U_\ve^{\rm P}+U_\ve^{\rm F}\) and then
periodically extend the resulting real multiplication
coefficient; we do not smoothly periodize \(V_{\rm P}\). Thus
the polynomial expansion in Eq.~\eqref{eq:poly-U-supp} and the Fourier split remain exact on the box.
We assume that this periodic-box realization is well posed on a
common invariant domain, that its generator \(\mathcal L_\ve\) is
skew-adjoint, and that its solutions satisfy the periodic Sobolev
compatibility required by the sampling and stencil estimates below.
We assume the standard observable-level domain estimate
\(\sup_{0<\ve\le1}|A_{\rm box}^\ve(T)-A^\ve(T)|
\le\delta_{{\rm dom},d,\delta}\le c_{\rm dom}\delta\), including the induced
initial-data and readout restriction errors; no dimension-uniform
localization theorem is invoked.

We record the only nonstandard point in the Hermite estimate. With
\(\Lambda_y=-\Delta_y+|y|^2\) and \(Q_K=I-\Pi_K\), Parseval gives
\begin{equation}
 \norm{\Lambda_y^rQ_KF}_2
 \le(2K+d)^{-(s-r)}\norm{\Lambda_y^sF}_2,
 \qquad 0\le r\le s,\quad r,s\in\mathbb R.
 \label{eq:tail}
\end{equation}
The oscillator graph norm and the ladder identities imply, for
\(\mathcal Z_j\in\{y_j,\partial_{y_j}\}\),
\begin{equation}
 \norm{\Pi_K\mathcal Z_jQ_KF}_2
 \le C_pK^{-(p-1/2)}\norm{\Lambda_y^pF}_2.
 \label{eq:tensor-boundary}
\end{equation}
More generally, if \(r=|\nu|\ge1\) is odd, then
\begin{equation}
 \norm{\Pi_Ky^\nu Q_KF}_2
 \le C_rK^{-(p-1/2)}
 \norm{\Lambda_y^{p+(r-1)/2}F}_2.
 \label{eq:tensor-monomial-boundary}
\end{equation}
Indeed, \(\norm{y^\nu G}_2\le
C_r\norm{\Lambda_y^{r/2}G}_2\), and Eq.~\eqref{eq:tail} applied to
\(G=Q_KF\) supplies the remaining factor.  The constants are independent
of \(d\) for fixed \(r\); all multi-index multiplicities are retained
below.
Define the Hermite--Galerkin approximation of this box solution by
\begin{equation*}
 \partial_tR_K^\ve=\Pi_K\mathcal L_\ve\Pi_KR_K^\ve,
 \qquad R_K^\ve(0)=\Pi_KR_{\rm box}^\ve(0).
\end{equation*}
For the error
\(e_K=R_K^\ve-\Pi_KR_{\rm box}^\ve\), skew-adjointness gives
\begin{equation*}
 \frac{\dd}{\dd t}\norm{e_K}_2
 \le\norm{\Pi_K\mathcal L_\ve Q_KR_{\rm box}^\ve}_2.
\end{equation*}
Write \(V_{\rm P}(x)=\sum_{|\alpha|\le D_{\rm P}}v_\alpha x^\alpha\),
\(\mathcal C_{{\rm P},d}:=\sum_\alpha|v_\alpha|\), and
\(\mathcal C_{{\rm F},d}^{(1)}
:=\norm{\nabla V_{\rm F}}_{L^\infty(\mathbb R^d;\ell^2)}\).
For any \(F\) in the displayed graph norms, let
\(\mathcal R_{\rm tr}\), \(\mathcal R_{\rm P}\), and
\(\mathcal R_{\rm F}\) denote the \(L^2\) norms of the contributions to
\(\Pi_K\mathcal L_\ve Q_KF\) from transport, the polynomial potential,
and the smooth non-polynomial potential, respectively.
The polynomial expansion in Eq.~\eqref{eq:poly-U-supp}, together with
Eqs.~\eqref{eq:tensor-boundary} and \eqref{eq:tensor-monomial-boundary}, gives
\begin{align*}
 \mathcal R_{\rm tr}
 \le C_pK^{-(p-1/2)}
 \sum_{j=1}^d\norm{\Lambda_y^p\partial_{x_j}F}_2,\qquad
 \mathcal R_{\rm P}
 \le C_{p,D_{\rm P}}\mathcal C_{{\rm P},d}
 K^{-(p-1/2)}\mathcal N_{m_*,\Omega_{x,\delta}}^{(d)}[F].
\end{align*}
An odd monomial of degree \(r\) requires weighted
\(y\)-regularity of order \(2p+r-1\) and a coefficient of degree at most
\(D_{\rm P}-r\);
the total is at most \(m_*\). Moreover,
\(\sum_{|\nu|=r,\,\nu\le\alpha}\binom{\alpha}{\nu}
=\binom{|\alpha|}{r}\), so no additional \(d\)-dependent combinatorial
constant is hidden here. For the smooth non-polynomial component, the full integral
representation, without Taylor expansion, gives
\(|U_\ve^{\rm F}(x,y)|\le
|y|\mathcal C_{{\rm F},d}^{(1)}\). Hence
\begin{equation*}
 \mathcal R_{\rm F}
 :=\norm{\Pi_KU_\ve^{\rm F}Q_KF}_2
 \le \mathcal C_{{\rm F},d}^{(1)}K^{-(p-1/2)}
 \norm{\Lambda_y^pF}_2,
\end{equation*}
uniformly in \(\ve\). Consequently,
\begin{equation*}
 \norm{\Pi_K\mathcal L_\ve Q_KR_{\rm box}^\ve}_2
 \le C_{p,D_{\rm P}}
 \bigl(1+\mathcal C_{{\rm P},d}
 +\mathcal C_{{\rm F},d}^{(1)}\bigr)K^{-(p-1/2)}
 \mathcal N_{m_*,\Omega_{x,\delta}}^{(d)}
 [R_{\rm box}^\ve],
\end{equation*}
where, for fixed \(p,D_{\rm P}\), the constant \(C_{p,D_{\rm P}}\)
is independent of \(d,\ve,K,\delta\). Time integration, followed by the tail
estimate \eqref{eq:tail} for \(Q_KR_{\rm box}^\ve\), gives
\begin{equation}
 \sup_{\substack{0<\ve\le1\\0\le t\le T}}
 \norm{R_{\rm box}^\ve(t)-R_K^\ve(t)}_2
 \le C_{{\rm HG},d}K^{-(p-1/2)},
 \label{eq:HG}
\end{equation}
where
\begin{align*}
 C_{{\rm HG},d}&=C_{p,D_{\rm P}}
 \Bigl[1+T\bigl(1+\mathcal C_{{\rm P},d}
 +\mathcal C_{{\rm F},d}^{(1)}\bigr)\Bigr]
 \mathfrak N_{m_*,T,{\rm box}}^{(d)}
\end{align*}
is independent of \(\ve,K,\delta\). Its dependence on \(d\) is
displayed through the potential data and the regularity norm.

For the spatial grid, take \(s_x(d)=\lfloor d/2\rfloor+1\), set
\[
 H_x^{s_x(d)}L_y^2:=
 H_{\rm per}^{s_x(d)}(\Omega_{x,\delta};L^2(\mathbb R_y^d)),
\]
and let
\(C_{{\rm st},d,\delta}\) be the associated vector-valued sampling and stencil
constant. At the cutoff \(K_{d,\delta}\) selected in
Eq.~\eqref{eq:Kselect} below, define
\begin{equation*}
 \mathfrak G_{d,\delta}:=
 \sup_{\substack{0<\ve\le1\\0\le t\le T}}
 \sum_{j=1}^d
 \norm{\partial_{x_j}^{r_x+1}\partial_{y_j}R_{K_{d,\delta}}^\ve(t)}
 _{H_x^{s_x(d)}L_y^2}.
\end{equation*}
We assume \(\mathfrak G_{d,\delta}<\infty\) for each fixed
\(d,\delta\), uniformly for \(0<\ve\le1\). This is an additional
spatial-regularity assumption: Eq.~\eqref{eq:mstar} alone does not imply
it when \(s_x(d)\) increases with \(d\). For the polynomial-in-\(d\)
conclusion, we additionally assume that
\(C_{{\rm st},d,\delta}\mathfrak G_{d,\delta}\) grows at most
polynomially in \(d\) at fixed \(\delta\).
Let \(\mathbf R_{\ve,h,K}^{\rm ex}\) be the exact sampled
semidiscrete evolution initialized by \(\mathcal S_hR_K^\ve(0)\), and
set
\(e_h=\mathcal S_hR_{K_{d,\delta}}^\ve
-\mathbf R_{\ve,h,K_{d,\delta}}^{\rm ex}\).
Let
\(\mathcal L_{\ve,h,K}^{\rm ex}:=-\ii
H_{\ve,h,K}^{\rm ex}\). Then \(e_h(0)=0\) and
\begin{equation*}
 \partial_te_h=\mathcal L_{\ve,h,K}^{\rm ex}e_h+\tau_h,
 \qquad
 \tau_h:=\partial_t(\mathcal S_hR_K^\ve)
 -\mathcal L_{\ve,h,K}^{\rm ex}\mathcal S_hR_K^\ve,
\end{equation*}
where the order-\(r_x\) stencil consistency and the vector-valued sampling
bound give
\(\norm{\tau_h}_2\le
C_{{\rm st},d,\delta}\mathfrak G_{d,\delta}h^{r_x}\). Since
\(\mathcal L_{\ve,h,K}^{\rm ex}\) is skew-Hermitian, the energy
estimate and time integration give
\begin{equation}
 \sup_{t\le T}\norm{e_h(t)}_2
 \le C_{{\rm st},d,\delta}T\mathfrak G_{d,\delta}h^{r_x}.
 \label{eq:xerr}
\end{equation}

Under the assumption that the absolute zeroth and
first Fourier moments are finite, define the weighted first-moment tail
\begin{equation}
 \mathfrak T_{{\rm F},d}(\Xi)=
 \begin{cases}
 \displaystyle(2\pi)^{-d}
 \int_{\norm{\xi}_\infty>\Xi}\norm{\xi}_1
 |\widehat V_{\rm F}(\xi)|\dd\xi,&\text{Fourier integral},\\[2mm]
 \displaystyle\sum_{\norm{\xi_q}_\infty>\Xi}
 |a_q|\norm{\xi_q}_1,&\text{periodic series}.
 \end{cases}
 \label{eq:Fourier-tail-supp}
\end{equation}
With \(\mathcal A_\xi=\xi\cdot y\), since
\(|z\operatorname{sinc}(\ve z/2)|\le|z|\) and
\(\norm{y_j\Pi_K}\le\sqrt{2K}\),
\begin{equation*}
 \norm{\mathcal A_\xi\operatorname{sinc}
 (\ve\mathcal A_\xi/2)\Pi_K}
 \le\sqrt{2K}\norm{\xi}_1,\qquad
 \eta_{\rm F}\le\sqrt{2K}\,\mathfrak T_{{\rm F},d}(\Xi)
 +\eta_{{\rm F},\rm quad}.
\end{equation*}
Here \(\eta_{{\rm F},\rm quad}\) is a uniform upper bound,
over \(0<\ve\le1\) and all grid points, on the retained-sector
operator-norm defect of the in-band positive-weight Fourier quadrature. For
a periodic series it includes any modes omitted inside the chosen band.

Under the \(O(1)\) input--readout assumption stated in
Sec.~\ref{sec:finite}, the following bounds can be chosen uniformly in \(d\); we retain
their possible \(d\)-dependence to make the resource accounting explicit.
Assume known bounds independent of \(\ve\) and of the
selected discretization,
\(s_R\le S_{R,d}\) and \(s_a\le S_{a,d}\) along the selected family, and
set
\begin{equation*}
 \mathcal S_d:=\max\{1,S_{R,d}S_{a,d},S_{a,d},\norm{\check a}_2\}
\end{equation*}
and fix \(0<c_K,c_{\rm F}<1\). We assume that the
positive-weight Fourier quadrature (or periodic-series truncation) can be
selected to satisfy the tail and in-band quadrature bounds below, with
stable absolute weights as in Eq.~\eqref{eq:stable}. Finiteness of the
Fourier moments guarantees only the existence of a tail cutoff for each
fixed \(d\); the in-band quadrature property remains an assumption, and
neither statement by itself implies polynomial dependence on \(d\). For the
polynomial-in-\(d\) resource conclusion, we therefore additionally assume
that, at each fixed \(\delta\), the selected bandwidth \(\Xi_{d,\delta}\),
mode count, quadrature and reversible-access costs, and input--output mesh
threshold grow at most polynomially in \(d\). Select
\begin{align}
 K_{d,\delta}
 &:=\min\{K\ge1:C_{{\rm HG},d}\norm{\check a}_2
 K^{-(p-1/2)}\le c_K\delta\},
 \label{eq:Kselect}\\
 \sqrt{2K_{d,\delta}}\,
 \mathfrak T_{{\rm F},d}(\Xi_{d,\delta})
 &\le c_{\rm F}\frac{\delta}{T\mathcal S_d},\qquad
 \eta_{{\rm F},\rm quad}=O\!\left(\frac{\delta}{T\mathcal S_d}\right).
 \nonumber
\end{align}
In particular,
\begin{equation}
 K_{d,\delta}=O\!\left[
 \left(1+\frac{C_{{\rm HG},d}\norm{\check a}_2}{\delta}
 \right)^{1/(p-1/2)}\right].
 \label{eq:Krate}
\end{equation}
For the selected quadrature, define
\begin{align}
N_{\mathcal Q,d,\delta}&:=|\mathcal Q_{d,\delta}|,
 &A_{0,d,\delta}&:=\sum_q|a_q|,
 \nonumber\\[-1mm]
 \mathcal V_{1,d,\delta}&:=\sum_q|a_q|\norm{\xi_q}_1,
 &A_{0,d,\delta}+\mathcal V_{1,d,\delta}&\le C_{V,d}.
 \nonumber\\[-1mm]
 B_{d,\delta}&:=\max_q\norm{\xi_q}_1,
 &B_{d,\delta}&\le d\Xi_{d,\delta}.
 \label{eq:stable}
\end{align}
The stability constant \(C_{V,d}\) is independent of
\(\delta\) and \(\ve\).
Using Eq.~\eqref{eq:buffer} at \(\ve=1\), with aggregate physical
operator tolerance \(O(\delta/(T\mathcal S_d))\), gives
\begin{align}
 L_{{\rm buf},d,\delta}=O\!\Big(&B_{d,\delta}^2
 +B_{d,\delta}\sqrt{K_{d,\delta}}
 +\log\!\left[1+
 \frac{T\mathcal S_d\max\{1,C_{V,d}\}
 (1+B_{d,\delta}\sqrt{K_{d,\delta}})}{\delta}\right]\Big).
 \label{eq:Lselect}
\end{align}
Thus the first two buffer terms are, in the worst case,
\(O(d^2\Xi_{d,\delta}^2+d\Xi_{d,\delta}\sqrt{K_{d,\delta}})\).
Set \(\widetilde K_{d,\delta}=K_{d,\delta}+L_{{\rm buf},d,\delta}\).
Substituting Eqs.~\eqref{eq:Atilde} and \eqref{eq:aF} into
Eq.~\eqref{eq:nF} with \(\delta_{\rm pot}=\Theta(\delta/(T\mathcal S_d))\)
gives
\begin{equation}
 n_{{\rm F},d,\delta}
 =O\!\left(
1+\sqrt{\widetilde K_{d,\delta}}B_{d,\delta}
 +\log\!\left[1+
 \frac{T\mathcal S_d\sqrt{\widetilde K_{d,\delta}}
 \mathcal V_{1,d,\delta}}{\delta}\right]\right).
 \label{eq:nF-selected}
\end{equation}
This is a uniform upper bound on the maximal degree for
\(0<\ve\le1\); the instance-dependent degree replaces
\(\sqrt{\widetilde K_{d,\delta}}B_{d,\delta}\) by
\(\ve\sqrt{\widetilde K_{d,\delta}}B_{d,\delta}\).

Let \(\kappa_{{\rm io},d,\delta}\) be the inverse-mesh threshold for input,
normalization, and observable quadrature. We assume their aggregate
observable-level error, including normalization factors, satisfies
\(\delta_{{\rm io},d,\delta}\le c_{\rm io}\delta\) whenever
\(h^{-1}\ge\kappa_{{\rm io},d,\delta}\).
Section~\ref{sec:io} illustrates this assumption for the factorized
Gaussian input and Gaussian readout. For more
general structured inputs and readouts, the existence and stated scaling
of this threshold are explicit assumptions. The
mesh choice
\begin{equation}
 h_{d,\delta}^{-1}\asymp\max\!\left\{1,
 \left(\frac{TS_{a,d}C_{{\rm st},d,\delta}\mathfrak G_{d,\delta}}{\delta}
 \right)^{1/r_x},\Xi_{d,\delta},\kappa_{{\rm io},d,\delta}\right\}
 \label{eq:hselect}
\end{equation}
controls the stencil error and prevents Fourier aliasing. Choose the QSVT
and finite-precision arithmetic tolerances so that
\begin{equation*}
 \eta_{\rm pot}=O\!\left(\frac{\delta}{T\mathcal S_d}\right).
\end{equation*}
The finite-section, QSVT, Fourier quadrature, and arithmetic defects enter
this single operator budget. Applying the triangle inequality along the
deterministic part of Eq.~\eqref{eq:proof-map}, followed by Duhamel's
identity, Eqs.~\eqref{eq:xerr} and \eqref{eq:HG}, and the
domain and input--output assumptions above give the deterministic
error bound
\begin{align}
 \sup_{0<\ve\le1}|A_{\rm WH}(T)-A^\ve(T)|
 \le{}&s_as_RT\eta_{\rm pot}
 +s_aC_{{\rm st},d,\delta}T\mathfrak G_{d,\delta}h_{d,\delta}^{r_x}
 \nonumber\\*[-1mm]
 &+\norm{\check a}_2C_{{\rm HG},d}K_{d,\delta}^{-(p-1/2)}
 +\delta_{{\rm io},d,\delta}+\delta_{{\rm dom},d,\delta}
 =O(\delta).
 \label{eq:ledger}
\end{align}
Here the last two terms contain the standard sampling, quadrature,
restriction, and periodization errors allocated above. When
\(V_{\rm F}=0\), take
\(N_{\mathcal Q,d,\delta}=B_{d,\delta}=\Xi_{d,\delta}=L_{{\rm buf},d,\delta}
=C_{V,d}=0\). This proves Lemma~\ref{lem:uniform-selection}.

\end{proof}

\section{Proof of the main theorem: uniform observable complexity}
\label{sec:proof3}

This section proves the main theorem of the Letter by combining
Lemma~\ref{lem:fixed-resources}, Lemma~\ref{lem:uniform-selection}, and
overlap estimation. The query, gate, and working-qubit bounds are given
explicitly below.

Together with Eq.~\eqref{eq:proof-map}, the final decomposition
\( |\widetilde A(T)-A^\ve(T)|
\le |\widetilde A(T)-A_{\rm WH}(T)|
+|A_{\rm WH}(T)-A^\ve(T)| \)
shows that Eq.~\eqref{eq:ledger} leaves only the normalized-state
preparation, Hamiltonian-simulation, and overlap-estimation errors to
allocate.

Set
\begin{equation*}
 \Lambda_{d,\delta}:=\max\{1,s_Rs_a/\delta\}.
\end{equation*}
If \(s_Rs_a=0\), return \(\widetilde A(T)=0\); the deterministic
bound in Eq.~\eqref{eq:ledger} still controls its error from the
continuum target. Otherwise, the preparation oracles of the main text are
\begin{equation*}
 O_R\ket0=\frac{\mathbf R_0}{s_R},\qquad
 O_a\ket0=\frac{\mathbf b_a}{s_a}.
\end{equation*}
Estimate the normalized overlap
\begin{equation*}
 \mu:=\left\langle\frac{\mathbf b_a}{s_a},
 e^{-\ii TH_{\rm WH}}\frac{\mathbf R_0}{s_R}\right\rangle,
 \qquad \widetilde A(T):=s_as_R\widetilde\mu,
\end{equation*}
where \(\widetilde\mu\) is the overlap estimate. Allocate the normalized
input-state, observable-state, Hamiltonian-simulation, and
overlap-estimation errors so that their sum is
\(O(\Lambda_{d,\delta}^{-1})\). After rescaling the
normalized overlap, their contribution is at most
\(O(s_Rs_a/\Lambda_{d,\delta})=O(\delta)\); combining this with
Eq.~\eqref{eq:ledger} proves the end-to-end error statement.

Evaluate the normalization in Eq.~\eqref{eq:alpha-fixed} at the parameters
chosen in Section~\ref{sec:proof2}:
\begin{equation*}
 \alpha_{{\rm WH},d}(\ve;\delta)
 :=\alpha_{\rm WH}(\ve;h_{d,\delta},K_{d,\delta},
 \mathcal Q_{d,\delta},L_{{\rm buf},d,\delta}).
\end{equation*}
The derivative estimate \eqref{eq:Gamma-d-bound} applies on
\(\Omega_{x,\delta}\).
Amplitude estimation and qubitized Hamiltonian simulation therefore use
\begin{align}
 N_H&=\wtO\!\left[
 \Lambda_{d,\delta}\left\{T\alpha_{{\rm WH},d}(\ve;\delta)
 +\log(1+\Lambda_{d,\delta})\right\}\right],
 \nonumber\\
 N_{\rm ladder}&=O\!\left(N_Hn_{{\rm F},d,\delta}\right),
 \qquad N_{\rm prep}=O(\Lambda_{d,\delta}),
 \label{eq:NH}
\end{align}
with constant success probability \cite{sm:Brassard2002}.

To turn these oracle queries into an end-to-end gate bound, let
\(\mathcal C_{{\rm out},d,\delta}\) be the cost of the outer
transport/polynomial/Fourier PREPARE--SELECT operations and let
\(\mathcal C_{{\rm mode},d,\delta}\) be the cost of one coordinate block
and its mode-dependent arithmetic. If \(\mathcal C_R,\mathcal C_a\) are
the two state-preparation costs, then
\begin{equation}
 G_{\rm tot}=\wtO\!\left(
 N_H[\mathcal C_{{\rm out},d,\delta}
 +n_{{\rm F},d,\delta}\mathcal C_{{\rm mode},d,\delta}]
 +\Lambda_{d,\delta}(\mathcal C_R+\mathcal C_a)\right).
 \label{eq:gates}
\end{equation}
Let \(\ell_{d,\delta}\) be the per-coordinate side length of the selected
periodic box and define the corresponding grid size by
\(M_{d,\delta}:=\lceil\ell_{d,\delta}/h_{d,\delta}\rceil\), with the usual
harmless adjustment so that the periodic mesh closes.
For the number
\(\mathsf P_{D_{\rm P},d}=\sum_{m=0}^{\qP}
\binom{2d+2m}{2m+1}\)
of polynomial shifted-diagonal LCU labels (using
\(\sum_{|\nu|=r}\prod_j(\nu_j+1)=\binom{2d+r-1}{r}\)),
the working-register size is
\begin{align}
 n_{\rm work}=O\!\Big(&d\log M_{d,\delta}
 +d\log\widetilde K_{d,\delta}+\log(1+N_{\mathcal Q,d,\delta})
 +\log(1+\mathsf P_{D_{\rm P},d})
 \nonumber\\[-1mm]
 &+\log(1+n_{{\rm F},d,\delta})+\log(1+\Lambda_{d,\delta})
 +b_{d,\delta}+w_{{\rm access},d,\delta}\Big),
 \label{eq:working-qubits}
\end{align}
where \(b_{d,\delta}\) is sufficient for potential-generator precision
\(O(\delta/(T\mathcal S_d))\) and full-Hamiltonian block-encoding precision
\(\eta_{\rm BE}\) with
\(T\eta_{\rm BE}=O(\Lambda_{d,\delta}^{-1})\). The costs
\(\mathcal C_R,\mathcal C_a\) target state-preparation error
\(O(\Lambda_{d,\delta}^{-1})\), and
\(w_{{\rm access},d,\delta}\) is the largest reversible workspace of all
data-access, PREPARE--SELECT, QSVT-phase, and state-preparation oracles.
Externally stored oracle data are not included.

For fixed \(T,\delta,p,r_x,D_{\rm P}\), these formulas are polynomial in
\(d\) if the derived constants and parameters themselves---in particular
\(C_{{\rm HG},d}\), \(C_{{\rm st},d,\delta}\),
\(\mathfrak G_{d,\delta}\), \(\kappa_{{\rm io},d,\delta}\),
\(S_{R,d}\), \(S_{a,d}\), \(\Xi_{d,\delta}\), \(C_{V,d}\),
\(N_{\mathcal Q,d,\delta}\), \(M_{d,\delta}\), the displayed
polynomial-derivative bounds, and
all PREPARE, SELECT, phase, arithmetic, workspace, and state-preparation
costs---grow at most polynomially in \(d\). This excludes
a hidden \(O(N_{\mathcal Q,d,\delta})\) table-loading cost unless
\(N_{\mathcal Q,d,\delta}\) is
itself polynomial. Without these assumptions, Eq.~\eqref{eq:NH} is an
oracle-query bound only. Finally, the instance-dependent Fourier degree
and the normalization bound \eqref{eq:alpha-fixed} contain only
nonnegative powers of
\(\ve\), while Eq.~\eqref{eq:nF-selected} and every selected
parameter are uniform for
\(0<\ve\le1\). This proves the main theorem of the Letter.

\section{Inputs and physical observables}
\label{sec:io}

\subsection{A positive uniformly regular input family}

Fix \(\sigma>1/2\), \(x_0,p_0\in\R^d\), and define
\begin{align}
 \psi_{x_0,p}^\ve(X)
 &=(2\pi\sigma^2)^{-d/4}
 e^{-|X-x_0|^2/(4\sigma^2)}e^{\ii p\cdot X/\ve},
 \nonumber\\
 \tau_\ve^2&=1-\frac{\ve^2}{4\sigma^2}>0,\qquad
 G_\ve(p)=(2\pi\tau_\ve^2)^{-d/2}
 e^{-|p-p_0|^2/(2\tau_\ve^2)},
 \nonumber\\
 \widehat\rho_0^\ve&=\int_{\R^d}G_\ve(p)
 \ket{\psi_{x_0,p}^\ve}\!\bra{\psi_{x_0,p}^\ve}\,\dd p.
 \label{eq:mixture}
\end{align}
The choice \(\sigma>1/2\) ensures, uniformly for \(0<\ve\le1\),
\begin{equation*}
 0<c_\sigma:=1-\frac1{4\sigma^2}\le\tau_\ve^2\le1,
 \qquad
 \norm{\psi_{x_0,p}^\ve}_2=1,\qquad
 \int_{\R^d}G_\ve(p)\,\dd p=1.
\end{equation*}
Indeed,
\begin{equation*}
 \langle f,\widehat\rho_0^\ve f\rangle
 =\int G_\ve(p)|\langle\psi_{x_0,p}^\ve,f\rangle|^2\,\dd p\ge0,
 \qquad
 \operatorname{Tr}\widehat\rho_0^\ve=\int G_\ve(p)\,\dd p=1.
\end{equation*}
Substituting \(X=x+\ve y/2\) and \(Y=x-\ve y/2\) into the
kernel of Eq.~\eqref{eq:mixture} and evaluating the Gaussian momentum
integral gives
\begin{equation*}
 R_0^\ve(x,y)=(2\pi\sigma^2)^{-d/2}
 e^{-|x-x_0|^2/(2\sigma^2)-|y|^2/2+\ii p_0\cdot y}.
\end{equation*}
Thus \(R_0^\ve\) is independent of \(\ve\) and
\begin{equation}
 \norm{R_0^\ve}_2=(2\sigma)^{-d/2},\qquad
 \mathcal N_m^{(d)}[R_0^\ve]<\infty\quad(m\ge0).
 \label{eq:R0norm}
\end{equation}
Since \(X-Y=\ve y\), its \(O(1)\) \(y\)-width is an \(O(\ve)\) coherence
length in the original density kernel.

For \(p_0=0\),
\(R_0^\ve\propto e^{-|x-x_0|^2/(2\sigma^2)}\Phi_{\mathbf0}(y)\). Hence
\begin{equation*}
 s_R^2=h^d\sum_{\mathbf i}|R_{\mathbf0}(x_{\mathbf i})|^2
 =\norm{R_0^\ve}_2^2+O(\delta)=O(1).
\end{equation*}

The amplitudes factor over coordinates. Grover--Rudolph conditional
rotations, with reversibly computed Gaussian interval masses, prepare the
normalized state with cost
\begin{equation*}
 \mathcal C_R=O\!\left(d\,\poly(\log M,\log K,\log(d/\delta))\right),
\end{equation*}
uniformly in \(\ve\) \cite{sm:GroverRudolph2002}. The mesh threshold
\(\kappa_{{\rm io},d,\delta}\) absorbs the difference between cell-mass and
point-sampled states.

A uniform \(L^2\)-norm bound is not automatic:
Eq.~\eqref{eq:change} gives
\begin{equation*}
 \norm{R^\ve}_2^2=\ve^{-d}\operatorname{Tr}[(\widehat\rho^\ve)^2].
\end{equation*}
A trace-one pure-state family extending to \(\ve\to0\) therefore has
\(\norm{R^\ve}_2=\ve^{-d/2}\), whereas Eq.~\eqref{eq:mixture} has purity
\((\ve/(2\sigma))^d\) and the uniform norm in Eq.~\eqref{eq:R0norm}.
This family verifies the uniform initial \(L^2\)-norm bound,
initial regularity, and structured-preparation assumptions. In one
dimension, for potentials satisfying the conditions of
Ref.~\cite[Proposition 4.1]{sm:FilbetGolse2026}, the full-space propagated
bound follows from Eq.~\eqref{eq:mstar}. The periodic-box bounds and the
higher-dimensional regularity and dimension-growth conditions are the
separate hypotheses stated in Sec.~\ref{sec:proof2}.

\subsection{Physical observables and readout}

We use the observable vector \(\mathbf b_a\) and pairing defined in
Eq.~\eqref{eq:observables}, with preparation and overlap estimation
as in Sec.~\ref{sec:proof3}.
For \(\sigma_a>0\), choose
\(c_a=(\pi\sigma_a^2)^{-d/4}\) and let
\(\check a(x,y)=c_ae^{-|x-x_a|^2/(2\sigma_a^2})
\Phi_{\mathbf0}(y)\). Then
\begin{equation*}
 s_a^2=h^d\sum_{\mathbf i}|c_ae^{-|x_{\mathbf i}-x_a|^2/(2\sigma_a^2)}|^2
 =1+O(\delta)=O(1),
\end{equation*}
and the same factorized preparation gives
\(\mathcal C_a=O(d\,\poly(\log M,\log K,\log(d/\delta)))\).
Local mass, momentum, and kinetic-energy densities obtained from
\(R^\ve(x,0)\) and its \(y\)-derivatives require additional trace
regularity. The main theorem of the Letter concerns fixed
\(L^2_{x,y}\) readout windows satisfying the stated sampling bounds.

\section{Numerical experiments}
\label{sec:numerics}

We present three one-dimensional experiments: Weyl profiles for polynomial
and smooth non-polynomial potentials, physical densities at small and order-one
\(\ve\) on the same grid, and rapidly oscillating densities produced by
a potential-only terminal propagation on the computed Hermite coefficients. The computations target the
full finite-\(\ve\) potential difference in Eq.~\eqref{eq:weyl-supp}, without taking its semiclassical limit.
For the Morse experiment, the smooth non-polynomial potential construction
uses a periodic extension, paired Fourier modes, and polynomial matrix
functions as specified below. The terminal pulse likewise uses paired
Fourier modes and a bounded polynomial matrix function, followed by
propagation of the finite Hermite state.
The local mass, momentum, and kinetic-energy densities are
\begin{equation}
 n^\ve(t,x)=R^\ve(t,x,0),\qquad
 j^\ve(t,x)=-\ii\partial_yR^\ve(t,x,0),\qquad
 E^\ve(t,x)=-\frac12\partial_y^2R^\ve(t,x,0).
 \label{eq:num-moments}
\end{equation}
Here \(j^\ve\) is momentum density, and \(E^\ve\) is the Weyl
kinetic-energy density, which need not be positive pointwise.
These local moments are classical diagnostics of the discretization;
the quantum observable theorem concerns the bounded readout windows
specified in Sec.~\ref{sec:io}.

Time integration for the autonomous stages uses the fourth-order symmetric composition
\begin{equation}
 S_4(\Delta t)=S_2(a\Delta t)S_2(b\Delta t)S_2(a\Delta t),\qquad
 S_2(\tau)=e^{-\ii\tau U/2}e^{-\ii\tau H_{\rm tr}}e^{-\ii\tau U/2},
 \quad a=\frac{1}{2-2^{1/3}},\quad b=1-2a.
 \label{eq:num-time}
\end{equation}
The Weyl--Hermite computations use scaled basis functions
\(\Phi_k^{(\ell)}(y)=\ell^{-1/2}\Phi_k(y/\ell)\), a centered spatial
derivative of order 28, and periodic numerical boundaries. FFT
diagonalization of the spatial difference matrix retains its
finite-difference symbol. The matrix subflows, including the polynomial
matrix functions in the Morse and terminal-pulse experiments, are evaluated classically; QSVT circuits are not
simulated in these experiments.
Independent reference solutions for the autonomous stages use Fourier discretization in both
\(x\) and \(y\), pointwise evaluation of the exact potential difference,
and Eq.~\eqref{eq:num-time}. The fine reference grids are specified below.
In the reference--coarse comparisons below, reference curves are green
and independently computed coarse values are blue open circles; every
coarse node in the displayed interval is shown.

\subsection{Weyl profiles for polynomial and smooth non-polynomial potentials}

We first consider
\begin{equation}
 V_{\rm P}(x)=\frac{x^2}{2}+\frac{x^4}{40},\qquad
 V_{\rm F}(x)=0.6(1-\cos x).
 \label{eq:num-profile-potentials}
\end{equation}
Their exact potential differences are
\begin{equation}
 U_\ve^{\rm P}(x,y)=(x+0.1x^3)y+0.025\ve^2xy^3,
 \qquad
 U_\ve^{\rm F}(x,y)=0.6\sin x\,y\,
 \operatorname{sinc}(\ve y/2),
 \label{eq:num-profile-differences}
\end{equation}
where \(\operatorname{sinc}(z)=\sin z/z\), with value one at zero.
Let \(G_\sigma(z)=(\sqrt{2\pi}\sigma)^{-1}
e^{-z^2/(2\sigma^2)}\), and define
\begin{align}
 n_0(x)&=\frac{1}{Z}\int_{\R}G_{\sigma_q}(q)
 [1+c\cos(\kappa_0q)]G_s(x-q)\,\dd q,
 \qquad Z=1+c e^{-\kappa_0^2\sigma_q^2/2},\nonumber\\
 R_0^\ve(x,y)&=n_0(x)
 \exp\!\left[-\frac12\left(\sigma_p^2+
 \frac{\ve^2}{4s^2}\right)y^2+\ii(2+0.3x)y\right],
 \label{eq:num-profile-initial}
\end{align}
with \(\sigma_q=1.15\), \(s=0.008\), \(\sigma_p=0.10\),
\(c=0.98\), and \(\kappa_0=40\). This is a positive Gaussian-packet
mixture with a spatial modulation whose frequency is independent of \(\ve\).

Figure~\ref{fig:num-weyl} shows the reference cross sections at \(T=0.15\)
for \(\ve=3\times10^{-3},10^{-3},10^{-5}\). All cases use the same
\(2048\times1024\) Fourier grid on
\([-8,8)\times[-80,80)\) and \(\Delta t=0.002\).
Both the \(y=0\) and \(x=0\) sections remain oscillatory, but their
visible oscillation scales do not shrink with decreasing \(\ve\).
Their amplitudes may change. This illustrates the relevant distinction:
uniform regularity permits oscillations whose scales remain controlled
as \(\ve\) decreases; it does not require a non-oscillatory profile.

\begin{figure}[!htbp]
 \centering
 \includegraphics[width=0.76\textwidth]{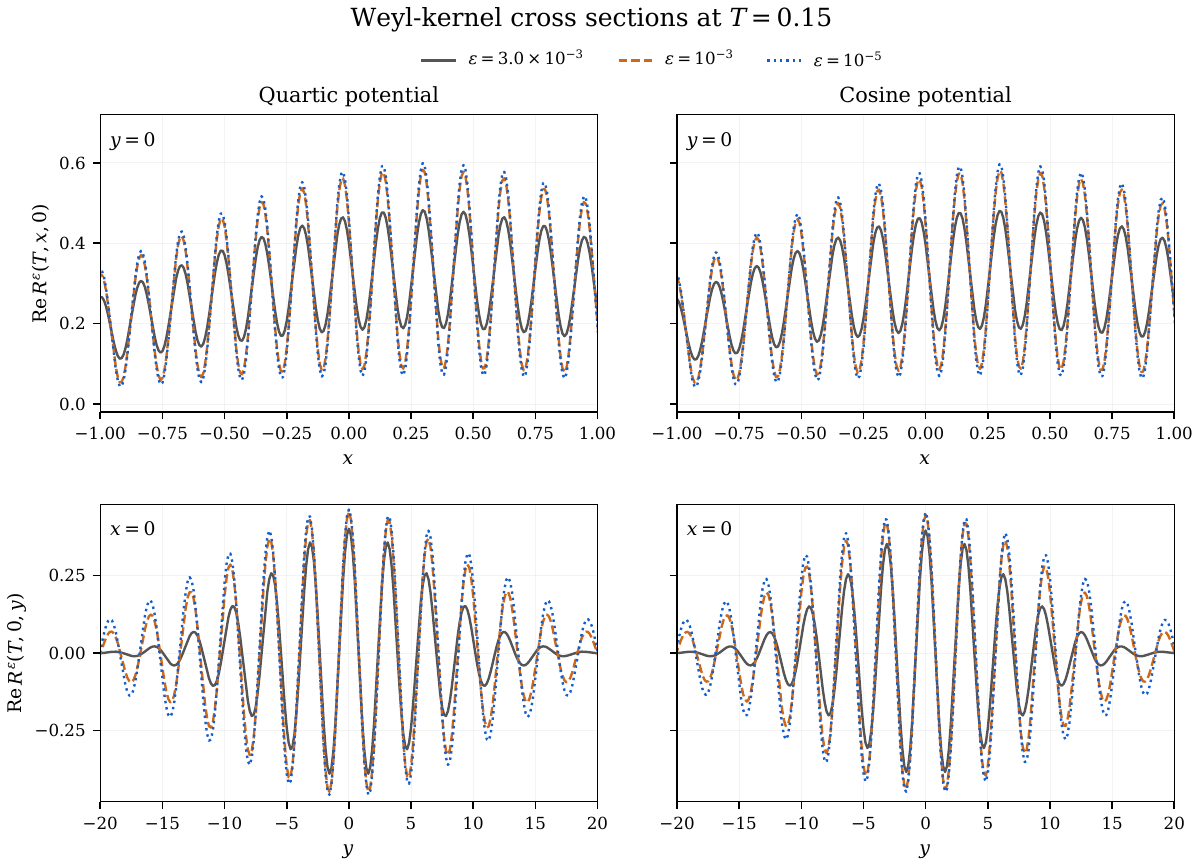}
 \caption{Weyl-kernel cross sections at \(T=0.15\) for the quartic potential
 (left) and cosine potential (right) in Eq.~\eqref{eq:num-profile-potentials}.
 The upper row shows \(\operatorname{Re}R^\ve(T,x,0)\); the lower row shows
 \(\operatorname{Re}R^\ve(T,0,y)\).
 Gray solid, orange dashed, and blue dotted curves correspond to
 \(\ve=3\times10^{-3},10^{-3},10^{-5}\), respectively.
 All curves are obtained on the same Fourier grid. The visible
 oscillation scales remain comparable as \(\ve\) decreases.}
 \label{fig:num-weyl}
\end{figure}
\FloatBarrier

\subsection{Morse potential: distinct regimes and computational resources}
\label{sec:num-morse}

We consider the Morse potential and the trace-one Gaussian initial state
\begin{equation}
 V(x)=20(1-e^{-0.16x})^2,\qquad
 R_0(x,y)=\frac{1}{\sqrt{2\pi}\,0.6}
 \exp\!\left[-\frac{(x-4)^2}{2(0.6)^2}-\frac{y^2}{2}\right].
 \label{eq:num-morse}
\end{equation}
The same initial Weyl profile is evolved to \(T=19\) for
\(\ve=10^{-4},10^{-3},10^{-2},0.1,0.2,0.5,1\).
We first compare the physical densities on a fixed grid and then select
parameters separately for each \(\ve\) to meet a common accuracy target.

\paragraph{Discretization and potential construction.}
The periodic spatial domain is \([-8,16)\), with nodes
\(x_i=-8+ih\), \(h=24/M\), and \(i=0,\ldots,M-1\).
We use \(M\) for the spatial grid size, as in the Letter.
We expand the Weyl profile in the orthonormal scaled Hermite functions
\begin{equation}
 \Phi_k^{(\ell)}(y)=\ell^{-1/2}\Phi_k(y/\ell)
 =\frac{\sqrt{2/3}}{\pi^{1/4}\sqrt{2^kk!}}
 H_k\!\left(\frac{2y}{3}\right)e^{-2y^2/9},
 \qquad \ell=\frac32,\quad k=0,\ldots,K-1.
 \label{eq:num-scaled-basis}
\end{equation}
Here \(H_k\) is the physicists' Hermite polynomial.
The initial coefficients are computed by 768-point Gauss--Legendre
quadrature on \([-40,40]\). All runs use the centered spatial derivative
of order 28 and the fourth-order time splitting in Eq.~\eqref{eq:num-time},
with \(\Delta t=0.01\).

To implement the smooth non-polynomial potential construction of the
Letter, we extend the Morse potential to a periodic function. Define
\begin{align}
 S(s)&=\begin{cases}
 0,&s\le0,\\
 \displaystyle\frac{e^{-1/s}}{e^{-1/s}+e^{-1/(1-s)}},&0<s<1,\\
 1,&s\ge1,
 \end{cases}
 \qquad
 \chi(z)=S\!\left(\frac{z+32}{8}\right)
         S\!\left(\frac{40-z}{8}\right),\nonumber\\
 V_{\rm ext}(z)&=20+W(z),\qquad
 W(z)=\chi(z)\left(20e^{-0.32z}-40e^{-0.16z}\right).
 \label{eq:num-morse-extension}
\end{align}
Thus \(V_{\rm ext}=V\) on \([-24,32]\), and \(W\) vanishes outside
\([-32,40]\). We periodize \(V_{\rm ext}\) on \([-36,44)\), of length
\(L=80\). The constant part, 20, cancels in the potential difference.
The remaining part is represented by the paired Fourier series
\begin{equation}
 W(z)\approx\sum_{q=-Q}^{Q}a_qe^{\ii\xi_qz},\qquad
 \xi_q=\frac{2\pi q}{L},\qquad
 a_q=\frac1L\int_{-36}^{44}W(z)e^{-\ii\xi_qz}\,\dd z.
 \label{eq:num-morse-fourier}
\end{equation}
The coefficients are computed by the periodic trapezoidal rule, using an
FFT with the spatial-origin phase included. We impose
\(a_{-q}=\overline{a_q}\) and omit the zero mode from the potential
difference. Here \(Q\) counts positive Fourier modes, so the paired
representation contains \(2Q\) nonzero modes. The extension and
coefficient quadrature are the same in all runs; the retained Fourier
cutoff is specified below.

Evaluating a matrix function after truncating the coordinate matrix
generally differs from projecting the full multiplication operator.
We therefore evaluate it on \(\widetilde K=K+32\) modes and compress
the result to the retained \(K\) modes. For the basis in
Eq.~\eqref{eq:num-scaled-basis}, the coordinate matrix is
\begin{equation}
 \Ymat_N^{(\ell)}=\ell\Ymat_N
 =\frac32\sum_{k=0}^{N-2}\sqrt{\frac{k+1}{2}}
 \bigl(\ket{k+1}\bra{k}+\ket{k}\bra{k+1}\bigr).
 \label{eq:num-scaled-coordinate}
\end{equation}
Thus \(\Ymat_N\) retains its unscaled definition from the Letter.
Let \(J:\mathbb C^K\to\mathbb C^{\widetilde K}\) be the zero-padding
embedding, and set
\begin{equation}
 B=\ell\sqrt{2\widetilde K},\qquad
 Z=\Ymat_{\widetilde K}^{(\ell)}/B,\qquad
 \beta_q=B\xi_q,\qquad \theta_q=\ve\beta_q/2\quad(q>0).
 \label{eq:num-morse-scaled-matrix}
\end{equation}
The factor \(\ell\) in \(\beta_q\) accounts for the scaled numerical
basis. Since \(\norm{\Ymat_{\widetilde K}^{(\ell)}}\le B\) and
\(\widetilde K\le224\) in every run, all shifted spectral points
\(x_i\pm\ve\lambda/2\), with
\(\lambda\in\sigma(\Ymat_{\widetilde K}^{(\ell)})\), lie in
\([-23.875,31.875]\subset[-24,32]\).
Before Fourier and polynomial approximation, the extension therefore
gives the same auxiliary matrix function as the original Morse
potential. The solution-mode truncation and the replacement of the
infinite coordinate operator by a finite matrix remain separate
approximations.

For each positive Fourier mode, we approximate
\(g_{q,\ve}(z)=z\operatorname{sinc}(\theta_qz)
=\sin(\theta_qz)/\theta_q\) on \([-1,1]\) by a real odd Chebyshev
polynomial \(p_{q,\ve}\). The truncated expansions are rescaled to
remain bounded by one, and their degrees satisfy
\begin{equation}
 2\sum_{q=1}^{Q}|a_q|\beta_q
 \norm{p_{q,\ve}-g_{q,\ve}}_{L^\infty([-1,1])}
 \le\delta_{\rm pot}.
 \label{eq:num-morse-poly-tol}
\end{equation}
Pairing the Fourier modes gives the real symmetric potential block
\begin{equation}
 U_{\ve,K}^{{\rm F},{\rm QSVT}}(x_i)
 =-2\sum_{q=1}^{Q}\operatorname{Im}(a_qe^{\ii\xi_qx_i})\,
 \beta_qJ^\dagger p_{q,\ve}(Z)J.
 \label{eq:num-morse-poly-block}
\end{equation}
These spatial blocks form
\(U_{\ve,h,K}^{{\rm F},{\rm QSVT}}
=\sum_i\ket i\bra i\otimes U_{\ve,K}^{{\rm F},{\rm QSVT}}(x_i)\),
using the notation of the Letter. The polynomials are evaluated
classically in the eigenbasis of the auxiliary coordinate matrix.
The compressed blocks are then diagonalized to apply the potential
subflows. This implements the polynomial matrix functions in the QSVT
construction of the Letter while retaining their full finite-\(\ve\)
dependence; quantum circuits are not simulated. The coarse evolutions
use Eq.~\eqref{eq:num-morse-poly-block} throughout.
In these classical tests, \(\delta_{\rm pot}\) controls polynomial
truncation and normalization in exact arithmetic, namely the QSVT
polynomial contribution to the finite-section implementation error.
Fourier quadrature, Fourier truncation, and floating-point errors are
separate. The corresponding block-encoding normalization is
\(\alpha_{\rm pot}^{\rm F}=2\sum_{q=1}^{Q}|a_q|\beta_q\), and the
maximum polynomial degree is
\(n_{\rm F}=\max_{1\le q\le Q}\deg p_{q,\ve}\).

\paragraph{Physical densities on a fixed grid.}
The independent reference uses the original Morse potential on a
\(1024\times2048\) Fourier grid over
\([-8,16)\times[-96,96)\), with \(\Delta t=0.01\).
It evaluates the finite-\(\ve\) potential difference pointwise and uses
the time splitting in Eq.~\eqref{eq:num-time}.
For each local density \(q\in\{n,j,E\}\), we report
\begin{equation}
 \mathcal E_\infty(q)=
 \frac{\max_i|q_{h,K}^\ve(T,x_i)-q_{\rm ref}^\ve(T,x_i)|}
 {\max_{x\in\mathcal G_{\rm ref}}|q_{\rm ref}^\ve(T,x)|},
 \label{eq:num-error}
\end{equation}
where the numerator uses all spatial nodes of the corresponding run,
and \(\mathcal G_{\rm ref}\) is the complete reference spatial grid.
We write \(e_q=\mathcal E_\infty(q)\); in particular, \(e_n\) is the
density-error notation used in Fig.~1 of the Letter.

Figure~\ref{fig:num-morse} compares \(\ve=10^{-4},0.2,1\) using the
fixed parameters
\begin{equation}
 M=256,\qquad h=0.09375,\qquad K=192,\qquad
 \widetilde K=224,\qquad Q=1536,\qquad\delta_{\rm pot}=10^{-8}.
 \label{eq:num-morse-grid}
\end{equation}
The same spatial grid and Hermite cutoff capture the three regimes.
Mass, momentum, and kinetic-energy densities have distinct profiles,
and the mass density changes substantially between small and order-one
\(\ve\). The small-\(\ve\) result therefore cannot replace the
\(\ve=1\) result for this initial Weyl profile. The displayed densities
are resolved by this grid; the large ratio \(h/\ve\) alone does not
demonstrate unresolved oscillations.

\paragraph{Relation to the benchmark in the Letter.}
Figure~1 of the Letter uses the same spatial mesh, retained and
auxiliary Hermite sizes, final time \(T=19\), and time step
\(\Delta t=0.01\). Its kernel panels show
\(g^\ve(X)=\operatorname{Re}\rho^\ve(T;X,0.98X-0.01)\) for
\(\ve=10^{-2}\) and \(5\times10^{-4}\), respectively.
The latter is an additional kernel-cut case beyond the seven values
used for the density comparison. The Weyl coordinate map gives
\[
 X_i=\frac{2x_i+0.01}{1.98},\qquad
 y_i=\frac{0.02X_i+0.01}{\ve},\qquad
 g_{h,K}^\ve(X_i)=\operatorname{Re}
 \sum_{k=0}^{K-1}c_{i,k}(T)\Phi_k^{(\ell)}(y_i).
\]
Here \(c_{i,k}(T)\) are the coefficients in the scaled numerical basis.
The displayed \(X\)-spacing is \(2h/1.98\), and all eleven
nodes in \([-1,0]\) are shown. The density panels use
\(n^\ve(T,x)=R^\ve(T,x,0)\); the interval \([-4,5]\) contains
96 nodes of this spatial mesh. The fixed-grid small-\(\ve\) density
error approaches \(1.225\times10^{-2}\), consistent with the rounded
annotation \(1.23\times10^{-2}\) in Fig.~\ref{fig:num-morse}.

\begin{figure}[!htbp]
 \centering
 \includegraphics[width=0.80\textwidth]{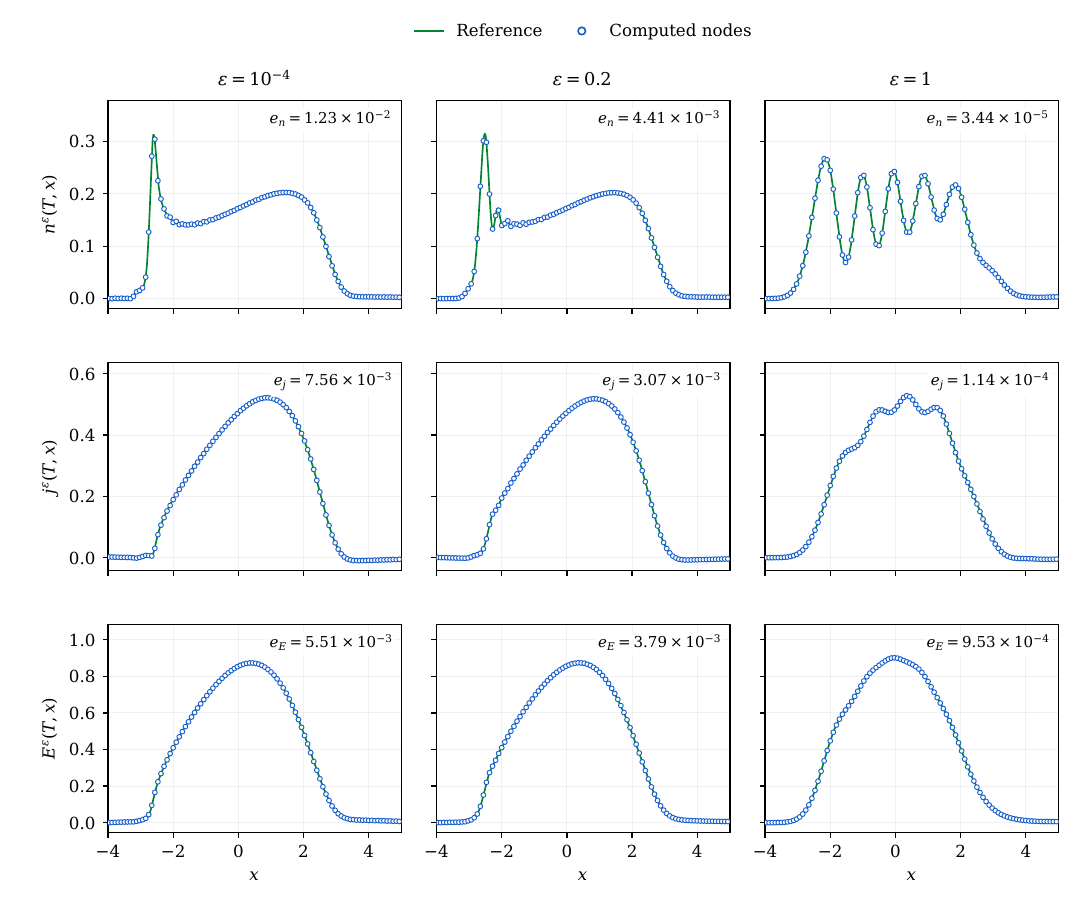}
 \caption{Morse-potential evolution at \(T=19\) with the fixed
 parameters in Eq.~\eqref{eq:num-morse-grid} and the potential construction
 in Eqs.~\eqref{eq:num-morse-extension}--\eqref{eq:num-morse-poly-block}.
 Columns correspond to \(\ve=10^{-4},0.2,1\); rows show mass density,
 momentum density, and Weyl kinetic-energy density.
 Green curves are the independent Fourier reference for the original
 Morse potential, and blue open circles show every computed spatial
 node in \([-4,5]\). The annotations give
 \(e_q=\mathcal E_\infty(q)\) over the full computational domain.
 All runs start from the unit-trace state in Eq.~\eqref{eq:num-morse}.}
 \label{fig:num-morse}
\end{figure}

\paragraph{Parameters at a common accuracy.}
To compare resources at a common accuracy, we choose \(M\), \(K\),
and \(Q\) separately for each \(\ve\), requiring
\begin{equation}
 \mathcal E_{\max}:=
 \max\{\mathcal E_\infty(n),\mathcal E_\infty(j),\mathcal E_\infty(E)\}
 \le5\times10^{-3}.
 \label{eq:num-common-accuracy}
\end{equation}
These runs use the same basis, extension, coefficient quadrature,
spatial difference formula, and time integrator as above, with
\(\widetilde K=K+32\) and \(\delta_{\rm pot}=10^{-6}\).
Each row of Table~\ref{tab:num-morse} is obtained by evolving from the
initial state with the listed parameters and checking all three
densities against the reference. These accuracy-selected runs are
distinct from the fixed-grid runs in Fig.~1 of the Letter and
Fig.~\ref{fig:num-morse}.

\begin{table}[!htbp]
 \centering
 \begin{tabular}{cccccc}
 \hline\hline
 \(\ve\)&\(M\)&\(K\)&\(Q\)&\(n_{\rm F}\)&\(\mathcal E_{\max}\)\\
 \hline
 \(10^{-4}\)&512&192&512&7&\(3.686\times10^{-3}\)\\
 \(10^{-3}\)&512&192&512&11&\(3.684\times10^{-3}\)\\
 \(10^{-2}\)&512&192&512&25&\(3.548\times10^{-3}\)\\
 \(0.1\)&512&192&512&99&\(3.274\times10^{-3}\)\\
 \(0.2\)&256&192&384&135&\(4.411\times10^{-3}\)\\
 \(0.5\)&256&96&256&163&\(3.728\times10^{-3}\)\\
 \(1\)&128&64&256&259&\(9.856\times10^{-4}\)\\
 \hline\hline
 \end{tabular}
 \caption{Parameters achieving the common accuracy target in
 Eq.~\eqref{eq:num-common-accuracy} at \(T=19\).
 All runs use the basis in Eq.~\eqref{eq:num-scaled-basis},
 \(h=24/M\), \(\widetilde K=K+32\), \(\Delta t=0.01\), and
 \(\delta_{\rm pot}=10^{-6}\). The Fourier cutoff \(Q\) counts
 positive modes; \(n_{\rm F}\) is the maximum degree of the potential
 polynomials. Errors include all computed spatial nodes.}
 \label{tab:num-morse}
\end{table}

All seven cases satisfy the same accuracy requirement. For \(\ve=1\),
the choices \(M=128\), \(K=64\), and \(Q=256\) give
\(n_{\rm F}=259\) and \(\mathcal E_{\max}=9.856\times10^{-4}\).
For the three tested values \(\ve=10^{-4},10^{-3},10^{-2}\), the common
choices \(M=512\), \(K=192\), and \(Q=512\) suffice, while the
polynomial degree decreases with \(\ve\). The table reports sufficient
parameters found by a finite search, without claiming minimal
computational cost.
\FloatBarrier

\subsection{Unresolved oscillations in a two-stage evolution}

We consider the von Neumann equation
\(\ii\ve\partial_t\widehat\rho=[\widehat H(t),\widehat\rho]\)
with the piecewise time-independent Hamiltonian
\begin{equation}
 \widehat H(t)=\begin{cases}
 -\dfrac{\ve^2}{2}\partial_x^2+V_b(x),&0\le t<T,\\[1mm]
 \dfrac{\Phi(x)}{\tau},&T\le t\le T+\tau,
 \end{cases}
 \qquad T=2,\quad\tau>0,
 \label{eq:num-two-stage-H}
\end{equation}
where
\begin{equation}
 V_b(x)=\frac{(x^2-1)^2}{4}+0.1x,\qquad
 \Phi(x)=-\frac{A}{\kappa}\sin(\kappa x),\qquad A=2,\quad\kappa=160.
 \label{eq:num-pulse-potentials}
\end{equation}
The first stage is evolution under a quartic potential. The second
stage contains only the sinusoidal potential and has no kinetic term.
Its amplitude is scaled by \(1/\tau\), so its total action is independent
of the duration \(\tau\). We denote the states at the switching and
final times by \(\widehat\rho^-=\widehat\rho(T)\) and
\(\widehat\rho^+=\widehat\rho(T+\tau)\), respectively. They satisfy
\begin{equation}
 \widehat\rho^+=e^{-\ii\Phi(\hat x)/\ve}\widehat\rho^-
 e^{\ii\Phi(\hat x)/\ve}.
 \label{eq:num-pulse}
\end{equation}
Thus Eq.~\eqref{eq:num-two-stage-H} realizes a terminal phase pulse
through a separate potential-only evolution.

\paragraph{Initial state and first-stage evolution.}
We take
\begin{equation}
 R_0(x,y)=\frac{1}{\sqrt{2\pi}\,0.65}
 \exp\!\left[-\frac{(x+1.2)^2}{2(0.65)^2}
             -\frac{0.85^2y^2}{2}+\ii0.6y\right],
 \label{eq:num-pulse-initial}
\end{equation}
with \(\ve=10^{-3}\), \(x\in[-8,8)\), \(M=256\),
\(h=0.0625\), and \(K=160\). Both stages use the Hermite basis in
Eq.~\eqref{eq:num-scaled-basis} and the scaled coordinate matrix in
Eq.~\eqref{eq:num-scaled-coordinate}.
For the first stage, the exact potential difference is
\(U_\ve=(x^3-x+0.1)y+\ve^2xy^3/4\).
The \(y^3\) multiplication matrix is formed on \(K+3\) modes and then
compressed to \(K\) modes, giving the exact retained multiplication
block. The 28th-order centered spatial difference and the fourth-order
splitting in Eq.~\eqref{eq:num-time}, with \(\Delta t=0.005\), yield
\(c_i^-\in\mathbb C^{160}\) at \(t=T\).

\paragraph{Second-stage evolution.}
The Fourier representation of \(\Phi\) has exactly two modes,
\(\Phi(x)=a_+e^{\ii\kappa x}+a_-e^{-\ii\kappa x}\), where
\(a_+=\ii A/(2\kappa)\) and \(a_-=\overline{a_+}\).
No Fourier quadrature or Fourier-mode truncation is required.
We construct the time-integrated second-stage generator on
\(\widetilde K=192\) auxiliary modes and then compress it to the
retained \(K=160\) modes. Let
\(J:\mathbb C^{160}\to\mathbb C^{192}\) be the zero-padding embedding
and set
\begin{equation}
 B=\ell\sqrt{2\widetilde K}=29.3938769\ldots,\qquad
 Z=\Ymat_{\widetilde K}^{(\ell)}/B,\qquad
 \beta=\kappa B,\quad\theta=\ve\beta/2=2.35151015\ldots.
 \label{eq:num-pulse-scaled-matrix}
\end{equation}
A real odd Chebyshev polynomial \(p_\theta\) approximates
\(g_\theta(z)=z\operatorname{sinc}(\theta z)=\sin(\theta z)/\theta\)
on \([-1,1]\) and is rescaled to remain bounded by one on this
interval. For this time-integrated generator, the normalization is
\(\alpha_\Phi=AB=58.7877538\ldots\), and the chosen degree is
\(n_{\rm F}=15\). The analytic truncation-plus-normalization bound is
\(\alpha_\Phi\norm{p_\theta-g_\theta}_\infty
\le4.242\times10^{-11}\), below the prescribed
\(\delta_{\rm pot}=10^{-10}\) polynomial budget for this integrated
generator. This is an exact-arithmetic polynomial bound; finite Hermite
projection and floating-point errors are separate.

Pairing the two Fourier modes gives the time-integrated Hermitian generator
\begin{align}
 U_{\Phi,K}^{\rm QSVT}(x_i)
 &=-AB\cos(\kappa x_i)J^\dagger p_\theta(Z)J,\nonumber\\
 c_i^+&=\exp\!\left[-\ii\tau\frac{U_{\Phi,K}^{\rm QSVT}(x_i)}{\tau}\right]c_i^-
 =\exp[-\ii U_{\Phi,K}^{\rm QSVT}(x_i)]c_i^-.
 \label{eq:num-pulse-coefficients}
\end{align}
The generator is compressed \emph{before} exponentiation.
The polynomial matrix and the \(K\)-dimensional exponential are
evaluated classically, as a realization of the polynomial target in
the Fourier--QSVT construction; no quantum circuit is simulated.
The computed post-pulse densities are extracted directly from the
propagated coefficients:
\begin{align}
 n_i^+&=\operatorname{Re}\sum_{k=0}^{K-1}c_{i,k}^+\Phi_k^{(\ell)}(0),
 \qquad
 j_i^+=\operatorname{Re}\sum_{k=0}^{K-1}
 (-\ii)c_{i,k}^+(\Phi_k^{(\ell)})'(0),\nonumber\\
 E_i^+&=\operatorname{Re}\sum_{k=0}^{K-1}
 \left(-\frac12\right)c_{i,k}^+(\Phi_k^{(\ell)})''(0).
 \label{eq:num-pulse-coefficient-readout}
\end{align}
The blue nodes in Fig.~\ref{fig:num-pulse} all come from
Eqs.~\eqref{eq:num-pulse-coefficients}--\eqref{eq:num-pulse-coefficient-readout}.
No analytic moment update is applied to those computed nodes.

\paragraph{Reference solution and nodal accuracy.}
The pre-pulse reference uses \(1024\times1536\) Fourier points on
\([-8,8)\times[-48,48)\), with \(\Delta t=0.005\).
For the second stage of this independent reference, we use the exact identities
\begin{equation}
 n^+=n^-,\qquad j^+=j^-+gn^-,\qquad
 E^+=E^-+gj^-+\tfrac12g^2n^-;
 \qquad g(x)=-\Phi'(x)=A\cos(\kappa x).
 \label{eq:num-pulse-moments}
\end{equation}
The reference pre-pulse moments are Fourier-interpolated to a 32768-node
grid on the full spatial domain before applying these identities.
The maximum absolute errors over all 256 computed nodes, divided by
the maximum absolute value of the respective full dense reference curve, are
\begin{equation}
 \mathcal E_\infty(n^+)=4.908\times10^{-4},\qquad
 \mathcal E_\infty(j^+)=6.874\times10^{-4},\qquad
 \mathcal E_\infty(E^+)=1.369\times10^{-3}.
 \label{eq:num-pulse-errors}
\end{equation}
Momentum contains the pulse wavelength \(2\pi/\kappa\simeq0.03927\),
and energy also contains \(\pi/\kappa\simeq0.01963\); both are
smaller than \(h=0.0625\), while density varies on a longer scale.
The finite Hermite propagation approximates the exact pulse, so the
reference identities need not hold exactly for the computed moments.

The two stages in Eq.~\eqref{eq:num-two-stage-H} use the polynomial and
smooth non-polynomial potential constructions, respectively.
The rapid oscillations are generated during the second stage, when
there is no spatial transport. We read out the densities at \(T+\tau\)
and perform no further evolution. The known sinusoidal potential is
evaluated at every coarse node. Accurate nodal values do not imply
that the complete oscillatory curves can be reconstructed from these
samples alone.

\begin{figure}[!htbp]
 \centering
 \includegraphics[width=0.54\textwidth]{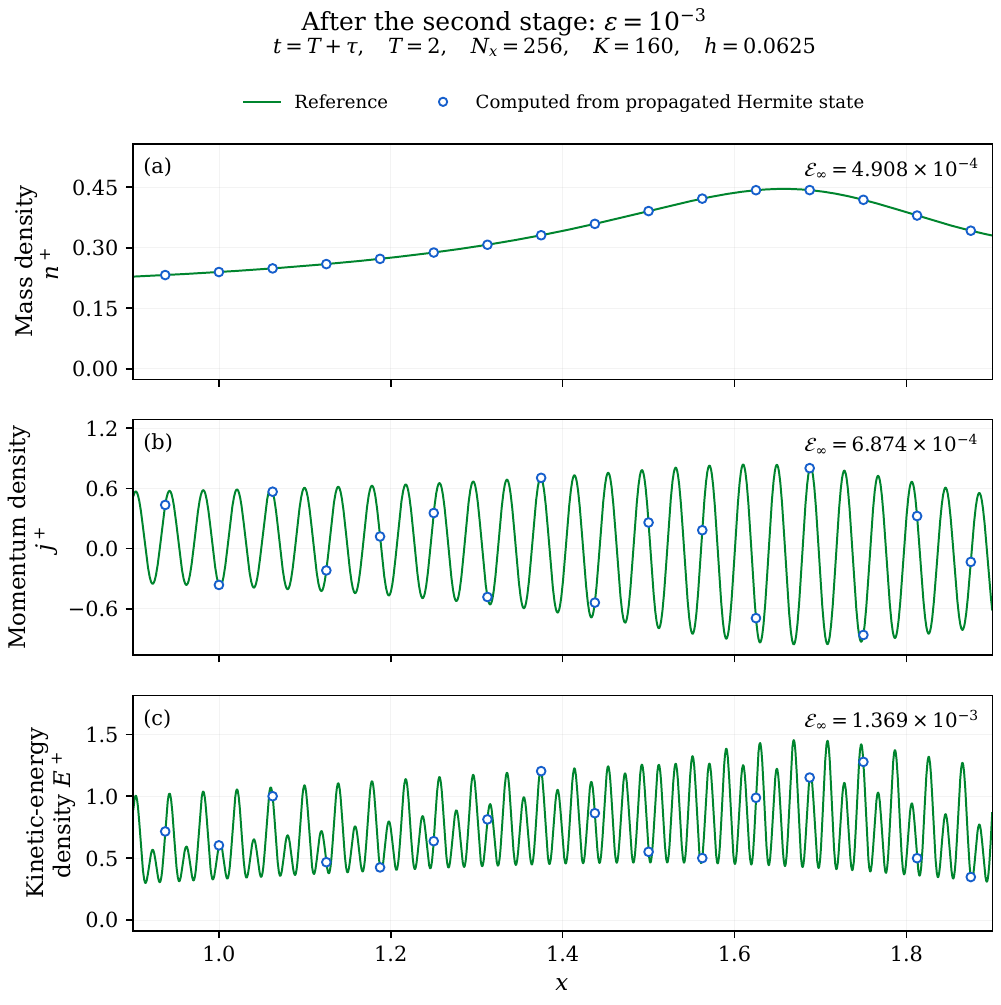}
 \caption{Mass, momentum, and Weyl kinetic-energy densities at the end
 of the two-stage evolution in Eq.~\eqref{eq:num-two-stage-H}, for
 \(\ve=10^{-3}\). Here \(T=2\) denotes the end of the first stage;
 the final readout is at \(T+\tau\).
 Green curves are the independent reference with the exact second-stage
 evolution applied. Blue open circles are read directly from the
 propagated \(K=160\) Hermite coefficients after exponentiating the
 compressed polynomial generator in Eq.~\eqref{eq:num-pulse-coefficients}.
 All 16 actual nodes in \(0.9\le x\le1.9\) are displayed.
 The spatial step \(h=0.0625\) exceeds the pulse wavelengths
 \(2\pi/\kappa\) and \(\pi/\kappa\). The errors in
 Eq.~\eqref{eq:num-pulse-errors} use all 256 nodes on the full domain.}
 \label{fig:num-pulse}
\end{figure}
\FloatBarrier

\end{document}